\documentclass[compsoc,conference]{IEEEtran}
\usepackage[english]{babel}
\usepackage[utf8]{inputenc}
\usepackage[T1]{fontenc}
\usepackage{amsfonts, amsmath, amsthm, amssymb}

\usepackage{arydshln}
\usepackage{balance}
\usepackage{enumerate}
\usepackage[shortlabels]{enumitem}
\usepackage[b]{esvect}
\usepackage{framed}
\usepackage[scr=boondoxo]{mathalfa}
\usepackage{mathpartir}
\usepackage{mathtools}
\usepackage{multirow}
\usepackage{pifont}
\usepackage{soul}
\usepackage{stmaryrd}
\usepackage{xcolor}
\usepackage{xifthen}
\usepackage{xspace}
\usepackage{tikz}
\usepackage{listings}
\usetikzlibrary{positioning}
\usepackage[pass]{geometry} %only for the appendix
\usepackage[textsize=tiny]{todonotes}
\usepackage{orcidlink}

\usepackage[colorlinks=true, urlcolor=blue]{hyperref}
\usepackage{microtype} 
\usepackage{comment}

\newboolean{draft}
\setboolean{draft}{false}
\newboolean{colored}
\setboolean{colored}{false}

\newboolean{fullversion}
\setboolean{fullversion}{true}
\newcommand{\ifproofs}[2]{\ifthenelse{\boolean{fullversion}}{#1}{#2}}
\newcommand{\iffull}[2]{\ifthenelse{\boolean{fullversion}}{#1}{#2}}

\theoremstyle{plain}
\newtheorem{theorem}{Theorem}
\newtheorem{lemma}{Lemma}
\newtheorem{proposition}{Proposition}
\newtheorem{corollary}{Corollary}

\theoremstyle{definition}
\newtheorem{definition}{Definition}
\newtheorem{assumption}{Assumption}
\newtheorem{example}{Example}
\newtheorem*{remark}{Remark}

\newcommand{\case}[1]{\underline{#1}}
\newcommand{\mypar}[1]{\smallskip\noindent\textbf{#1.}}

\ifthenelse{\boolean{draft}}{%draft
\newcommand{\fix}[1]{\footnote{#1}}
}{% final version
\newcommand{\fix}[1]{}
}
\newcommand{\namedfix}[2]{\fix{\textbf{#1}: #2}}
\newcommand{\christoph}[1]{\namedfix{CS}{\textcolor{blue}{#1}}}

\ifthenelse{\boolean{colored}}{ %colored
\newcommand{\newcr}[1]{\textcolor{blue}{#1}}    % new stuff for camera-ready version
\newcommand{\new}[1]{\textcolor{blue}{#1}}

\newenvironment{newtext}{\color{blue}}{}
}{ %final version
\newcommand{\newcr}[1]{#1}    % new stuff for camera-ready version
\newcommand{\new}[1]{#1}

\newenvironment{newtext}{}{}
}

\newcommand{\id}[1]{\mathit{#1}} 		% identifiers, variables, etc
\newcommand{\kw}[1]{\mathsf{#1}}		% keywords, constructors, etc
\newcommand{\perm}[1]{\mathbf{#1}}	% I/O permissions
\newcommand{\op}[1]{\mathtt{#1}}		% I/O operation

\newcommand{\vect}[1]{\overline{#1}}    	% vectors of terms etc.

\newcommand{\fun}{\rightarrow}     		% function 
\newcommand{\kleene}[1]{#1^{\star}} 	% sequences / lists 
\newcommand{\img}{\mathit{img}}		% image of a function

\newcommand{\N}{\mathbb{N}}
\newcommand{\B}{\mathbb{B}}        		% set of bytes
\newcommand{\bytestrings}{\kleene{\B}}	
	
\newcommand{\Calg}{\mathcal{B}}  

\newcommand{\msrsys}{\mathcal{R}}                  % generic/global MSR system
\newcommand{\RR}{\msrsys}				 % shorthand
\newcommand{\msrsub}[1]{\msrsys_\mathsf{#1}}
\newcommand{\Renv}{\msrsub{env}}
\newcommand{\MD}{\mathit{MD}}
\newcommand{\FreshR}{\mathit{Fresh}}
\newcommand{\Rbuf}{\msrsub{io}}
\newcommand{\Rintf}{\msrsub{intf}}
\newcommand{\Rrole}[1]{\msrsub{role}^{#1}}

\newcommand{\Alice}{\id{Alice}}

\newcommand{\facts}{\mathrm{facts}}
\newcommand{\fact}[1]{\mathsf{#1}}

\newcommand{\Frf}{\fact{Fr}}

\newcommand{\Inf}{\fact{in}}
\newcommand{\Outf}{\fact{out}}
\newcommand{\Eq}{\mathsf{E}}
\newcommand{\Fre}{\fact{Fresh}}
\newcommand{\Setupf}{\fact{Setup}}

\newcommand{\pkf}{\fact{pk}}		% no bang!
\newcommand{\skf}{\fact{sk}}		% no bang!
\newcommand{\K}{\mathsf{K}}		% no bang!

\newcommand{\knowsf}{\fact{K}}

\newcommand{\encSYM}{\mathsf{enc}}
\newcommand{\enc}[2]{\encSYM(#1,#2)}

\newcommand{\decSYM}{\mathsf{dec}}
\newcommand{\dec}[2]{\decSYM(#1,#2)}
\newcommand{\pairSYM}[1]{\mathsf{pair}}

\newcommand{\tuple}[1]{\langle #1 \rangle}

\newcommand{\getmsgSYM}[1]{\mathsf{getmsg}}

\newcommand{\hashSYM}{\mathsf{hash}}
\newcommand{\hash}[1]{\hashSYM(#1)}
\newcommand{\Calgify}[1]{#1^{\Calg}}
\newcommand{\encBsym}{\Calgify{\encSYM}}
\newcommand{\encB}[2]{\encBsym(#1,#2)}
\newcommand{\decBsym}{\Calgify{\decSYM}}
\newcommand{\decB}[2]{\decBsym(#1,#2)}

\newcommand{\udis}{\uplus}
\newcommand{\bigudis}{\biguplus}

\newcommand{\interl}{\mathrel{|||}}
\newcommand{\sync}[1]{\mathrel{\parallel_{#1}}}

\newcommand{\med}{\pi}
\newcommand{\medint}{\new{\pi_{\mathit{int}}}}
\newcommand{\medext}{\new{\pi_{\mathit{ext}}}}
\newcommand{\medextp}{\new{\pi'_{\mathit{ext}}}}

\newcommand{\sigsub}[1]{\Sigma_{\mathsf{#1}}}
\newcommand{\sigfacts}{\sigsub{facts}}
\newcommand{\siglin}{\sigsub{lin}}
\newcommand{\sigper}{\sigsub{per}}
\newcommand{\sigact}{\sigsub{act}}
\newcommand{\sigenv}{\sigsub{env}}
\newcommand{\sigstate}[1]{\sigsub{state}^{#1}}
\newcommand{\sigin}{\sigsub{in}}
\newcommand{\sigout}{\sigsub{out}}
\newcommand{\sigbuf}[1]{\sigsub{buf}^{#1}}
\newcommand{\sigrole}[1]{\sigsub{role}^{#1}}

\newcommand{\fac}{\mathcal{F}}
\newcommand{\faclin}{\mathcal{F}_{\mathsf{lin}}}
\newcommand{\facper}{\mathcal{F}_{\mathsf{per}}}
\newcommand{\faceq}{\Sigma_{\mathsf{eq}}}
\newcommand{\feq}{\phi_{\mathsf{eq}}}
\newcommand{\req}{\mathsf{R}_{\mathsf{eq}}}

\newcommand{\freshtype}{\mathit{fresh}}
\newcommand{\pubtype}{\mathit{pub}}
\newcommand{\names}{\mathcal{N}}       % names = fresh \cup pub
\newcommand{\vartype}{\mathcal{V}}
\newcommand{\rid}{\mathit{rid}}

\newcommand{\Terms}{\mathcal{T}}
\newcommand{\TermsSig}{\Terms_{\Sigma}}
\newcommand{\TermsFull}{\TermsSig(\names\cup\vartype)}
\newcommand{\GTerms}{\TermsSig(\names)}

\newcommand{\Msgs}{\mathcal{M}}
\newcommand{\Msg}{m}
\newcommand{\GT}{\mathscr{G}}

\newcommand{\traces}{\mathrm{Tr}}
\newcommand{\tamtraces}{\traces_{\mathsf{t}}}
\newcommand{\tamtracesp}{\traces'}
\newcommand{\tamtraceseq}{\traces_{\mathsf{t}}}
\newcommand{\tamtracespeq}{\traces'_{eq}}
\newcommand{\tamtraceseqq}{\traces_{\mathsf{t},eq}}

\newcommand{\Phieq}{\Phi_{\mathsf{eq}}}

\newcommand{\tracePre}{\mathrel{\preccurlyeq}} 
\newcommand{\tamtracePre}{\mathrel{\tracePre_{\mathrm{t}}}} 
\newcommand{\tamtracepPre}{\mathrel{\tracePre'}} 

\newcommand{\redms}[2]{\stackrel{\;#2\;}{\Longrightarrow}_{#1}}
\newcommand{\redmseq}[2]{\stackrel{\;#2\;}{\Longrightarrow}_{#1,\mathsf{eq}}}

\newcommand{\redev}[1]{\xrightarrow{\;#1\;}}
\newcommand{\trans}[1]{\xrightarrow{#1}}

\newcommand{\EE}{\mathscr{E}}
\newcommand{\CC}{\mathscr{C}}

\newcommand{\realEnv}{\EE}

\renewcommand{\S}{\mathcal{S}}
\newcommand{\E}{\mathcal{E}}
\newcommand{\G}{\mathcal{G}}
\newcommand{\U}{\mathcal{U}}

\newcommand{\vars}{\mathrm{vars}}

\newcommand{\evs}[1]{e^s_{#1}}
\newcommand{\eve}[1]{e^e_{#1}}

\newcommand{\pip}{\tilde{\pi}}

\newcommand{\coll}{\mathcal{C}}

\newcommand{\PP}{\psi}
\newcommand{\prog}{c}

\newcommand{\rew}[1]{\xrightarrow{#1}}

\newcommand{\IOSpecRole}[1]{P_{#1}}

\newcommand{\Token}[1]{token(#1)}

\newcommand{\msetopify}[1]{#1^\mathsf{m}}
\newcommand{\enumM}[1]{[#1]}
\newcommand{\emptyM}{\enumM{}}
\newcommand{\inM}{\msetopify{\in}}
\newcommand{\cupM}{\msetopify{\cup}}
\newcommand{\capM}{\msetopify{\cap}}
\newcommand{\subM}{\msetopify{\subseteq}}
\newcommand{\setminusM}{\msetopify{\setminus}}
\newcommand{\setM}{\mathrm{set}}

\newcommand{\skipE}{\mathsf{skip}}

\renewcommand{\ll}{\mathscr{l}}
\renewcommand{\aa}{\mathscr{a}}
\newcommand{\rr}{\mathscr{r}}

\newcommand{\multileft}{\{\hspace{-0.2em}|}
\newcommand{\multiright}{|\hspace{-0.2em}\}}

\newcommand{\aead}{\mathrm{aead}}
\newcommand{\h}{\mathrm{h}}
\newcommand{\kdf}{\mathrm{kdf}}

\newcommand{\IR}[1]{{#1}_\mathit{IR}}
\newcommand{\RI}[1]{{#1}_\mathit{RI}}
\newcommand{\kIR}{\IR{k}}
\newcommand{\kRI}{\RI{k}}
\newcommand{\nIR}{\IR{n}}
\newcommand{\nRI}{\RI{n}}

\lstdefinelanguage{gobra}{
  language=go,
  sensitive=true,
  morecomment=[l]{//},
  morecomment=[s]{/*}{*/},
  morekeywords=[1]{ %% Keywords of the programming language subset
    pred, implements, ghost, set
  },
  morekeywords=[2]{ %% Keywords of the specification language subset
    requires, ensures, invariant, req, ens, pure,  unfolding, in, forall, acc
  },
  morekeywords=[3]{ %% Keywords of the proof language subset
    fold, unfold,
    assume, assert, inhale, exhale
  },
  basicstyle={\ttfamily\footnotesize},
  commentstyle={\color[HTML]{747678}\textit},
  keywordstyle={[1]\color[HTML]{0005FF}},%\bfseries
  keywordstyle={[2]\color[HTML]{CC5500}},%\bfseries
  keywordstyle={[3]\color[HTML]{EC008C}},%\bfseries
  mathescape=true,
  moredelim=**[is][\normalfont\itshape]{'}{'}
}

\lstnewenvironment{gobra}[1][]{
  \lstset{
    language=gobra,
    floatplacement={tbp},
    captionpos=b,
    frame=lines,
    numbers=left, % note that several listings in section "introduction" rely on the line numbers
    numberblanklines=false,
    xleftmargin=8pt,
    xrightmargin=8pt,
    numberstyle=\scriptsize,
    breaklines=true,
    #1
  }
}{}

\def\code{%
    \lstinline[language=gobra,basicstyle=\ttfamily]}

\title{Sound Verification of Security Protocols:\\From Design to Interoperable Implementations \iffull{(extended version)}{}\vspace{-.25em}}

\author{
\IEEEauthorblockN{%
	Linard Arquint\IEEEauthorrefmark{1} \orcidlink{0000-0002-6230-8014},
	Felix A.\ Wolf\IEEEauthorrefmark{1} \orcidlink{0000-0002-8573-2387},
	Joseph Lallemand\IEEEauthorrefmark{2},
	Ralf Sasse\IEEEauthorrefmark{1} \orcidlink{0000-0002-5632-6099},\\
	Christoph Sprenger\IEEEauthorrefmark{1} \orcidlink{0000-0003-2941-5165},
	Sven N.\ Wiesner\IEEEauthorrefmark{1} \orcidlink{0000-0002-0468-5196},
	David Basin\IEEEauthorrefmark{1} \orcidlink{0000-0003-2952-939X}, and
	Peter M{\"{u}}ller\IEEEauthorrefmark{1} \orcidlink{0000-0001-7001-2566}
}
\IEEEauthorblockA{\IEEEauthorrefmark{1}\textit{Department of Computer Science, ETH Zurich, Switzerland}}
\IEEEauthorblockA{\IEEEauthorrefmark{2}\textit{Univ Rennes, CNRS, IRISA, France}}
\IEEEauthorblockA{\texttt{\{linard.arquint, felix.wolf, ralf.sasse, sprenger, sven.wiesner,}\\\texttt{basin, peter.mueller\}@inf.ethz.ch, joseph.lallemand@irisa.fr}}
}

\begin{document}

% to also have the page number on the first page
\thispagestyle{plain}

\maketitle

\begin{abstract}
\new{We provide a framework consisting of tools and metatheorems for the 
end-to-end verification of security protocols, 
which bridges the gap between automated protocol verification 
and code-level proofs.
We automatically translate a
Tamarin protocol model 
into a set of I/O specifications expressed in separation logic. 
Each such specification describes a protocol role's intended I/O behavior 
against which the role's implementation is then verified.
Our soundness result guarantees that the verified implementation
inherits all security (trace) properties proved for the Tamarin model.
Our framework thus enables us to leverage the substantial body of prior verification work in Tamarin 
to verify new and existing implementations. 
The possibility to use any separation logic code verifier provides flexibility regarding 
the target language.
To validate our approach and show that it scales to real-world protocols, 
we verify a substantial part of the official Go implementation of the WireGuard VPN key 
exchange protocol.}
\end{abstract}

\begin{IEEEkeywords}
Protocol verification, Symbolic security, Automated verification, Tamarin, Separation logic, Implementation.
\end{IEEEkeywords}

% !TEX root = main.tex

\section{Introduction}
\label{sec:introduction}

Security protocols are central to securing communication
and distributed computation and, by nature, they are often employed in critical applications.
Unfortunately, as history amply demonstrates, they are notoriously difficult to get right,
and their flaws can be a source of devastating attacks.
Hence the importance of their formal modeling and verification. 

Over the past decades, expressive and highly automated security protocol verifiers
have been developed, including the two state-of-the-art tools 
Tamarin~\cite{SchmidtMCB12,tamarin13} and ProVerif~\cite{DBLP:conf/csfw/Blanchet01}, \new{which 
have been used to analyze real-world protocols
such as TLS~\cite{DBLP:conf/sp/BhargavanBK17}, 5G~\cite{DBLP:conf/ccs/BasinDHRSS18}, and EMV~\cite{DBLP:conf/sp/BasinST21}.}
These tools build on a model of cryptographic protocols
called the \emph{symbolic} or \emph{Dolev-Yao model}, where cryptographic primitives are idealized,
protocols are modeled by process algebras or rewriting systems,
and the attacker is an abstract entity controlling the network and manipulating messages represented as terms.
However,
the protocol verified is a highly abstract version of the actual protocol executed
and there is \textit{a priori} no formal link between these two versions.
The problem of verifying the security of protocol implementations has been studied before,
but the solutions usually come with severe limitations, \new{the most important of which we highlight here. Section~\ref{sec:relatedWork} contains a detailed overview of related work.}

\looseness=-1
Many  existing approaches are based on code generation or model extraction, and
either extract executable code from a relatively abstract verified model (e.g.~\cite{DBLP:conf/fps/Modesti15,DBLP:conf/birthday/AlmousaMV15,DBLP:journals/fac/SistoCAP18}) or conversely extract a model from the code for verification
(e.g.~\cite{DBLP:journals/toplas/BhargavanFGT08}).
Another recent approach, DY$^\star$~\cite{BBHHKSW21},
provides a framework to write an executable implementation in F$^\star$ and obtain a corresponding model that can be reasoned about symbolically, also in F$^\star$.
Other methods adopt an approach, where security is proved in a computational model (e.g.~\cite{DBLP:conf/post/CadeB13, DBLP:conf/sp/Delignat-Lavaud17, DBLP:conf/sp/Delignat-Lavaud21}).
These models are more precise and thus give stronger guarantees than symbolic ones. However, their proofs are very difficult to automate.

\looseness=-1
These previous approaches are usually tied to a specific implementation language,
like ML, F\#, or Java dialects, and they are difficult to extend to other languages.
They are therefore ill-suited to verifying pre-existing implementations, especially
when used for code extraction.
In addition, the extraction mechanisms used are not always proved correct or even formalized, which 
weakens the guarantees for the resulting code.
Moreover, in many cases, the security proof is tailored to the implementation considered rather than constructed at an abstract level by a standard security protocol verifier such as Tamarin or ProVerif.  Hence one can neither leverage
these tools' automation capabilities nor the
substantial prior work invested in security protocol proofs using them.

\subsubsection*{Our approach}
We propose a novel approach to end-to-end verified security protocol implementations.
Our approach leverages the combined power of state-of-the-art security protocol verifiers
and source code verifiers.
This provides abstract, concise, and expressive security protocol specifications on the modeling side
and flexibility and versatility on the implementation side.

More precisely, our approach bridges abstract security protocol models expressed in 
Tamarin as multi-set rewriting systems 
with concrete program specifications expressed as I/O specifications (in a dialect of separation logic~\cite{DBLP:conf/lics/Reynolds02}), 
against which implementations can be verified.
Its technical core is a procedure, \new{implemented in an associated tool,} that translates Tamarin models into I/O specifications along with a soundness
proof, stating that an implementation satisfying the I/O specifications refines the abstract model 
in terms of trace inclusion. As a result, any \emph{trace property} proved for the abstract model using Tamarin, 
including standard security protocol properties such as secrecy and authentication, also holds for
the implementation.
 
Our approach provides a modular and flexible way to verify security protocol implementations.
On the model verification side, \new{we can leverage} Tamarin's proof automation capabilities to prove protocols secure.
Moreover, we can prove a \new{given} protocol's security once in Tamarin, and reuse this proof to verify
multiple \new{implementations of this protocol}, rather than having to produce a custom security proof for each implementation. 
In fact, numerous complex, real-world protocols have been analyzed using Tamarin over the years. Using our method, 
this substantial body of prior work can be exploited to verify implementations.

On the code verification side, \new{the I/O specifications we produce} can be encoded in many existing verifiers that support separation logic. \new{Our tool currently generates I/O specifications for the Go code verifier Gobra~\cite{Gobra} and for the Java code verifier VeriFast~\cite{DBLP:conf/nfm/JacobsSPVPP11}, which we respectively use for our case study and for our running example. It could easily be extended to other verifiers supporting I/O specifications such as Nagini~\cite{DBLP:conf/cav/Eilers018} for Python code.}
In addition, the requirements for adding other code verifiers based on separation logic to our arsenal are low: they need \new{to} only support abstract predicates to encode I/O specifications and to guarantee that successful verification implies trace inclusion between the I/O traces of the program and those of its I/O specification.

\looseness=-1
\new{We establish our central soundness result relating Tamarin models via I/O specifications to implementations.}
This result follows
a methodology inspired by the Igloo framework~\cite{Igloo}, which provides a series of generic steps that 
gradually transform an abstract model into an I/O specification, and requires establishing a refinement relation 
between each successive pair of steps.
We take similar steps and prove these refinements \emph{once and for all} starting from a generic Tamarin system, so that
our method can be applied to obtain an I/O specification from any Tamarin protocol model (under some mild syntactic assumptions) without any additional proof.

\subsubsection*{Our contributions}
We summarize our contributions as follows.
First, we design a framework for the end-to-end verification of security protocol implementations. This consists of a procedure \new{and an associated tool} to extract \new{I/O specifications} from a Tamarin model, which can be verified on implementation code. 
\new{Our soundness result ensures that the implementation inherits all properties proven in Tamarin.}

\new{Second, we propose a novel approach to relate the I/O specifications' symbolic terms to the code's bytestring messages. We parameterize the code verifiers' semantics with an abstraction function, instantiate it to a message abstraction function, and identify assumptions and proof obligations to verify that the code correctly implements terms as bytestrings.}

\looseness=-1
Finally,  to validate our approach, we perform a substantial case study on a \new{complex, real-world} protocol: the WireGuard key exchange, which is part of the widely-used WireGuard VPN \new{in the Linux kernel. 
We verify a part of} the official Go implementation of WireGuard which is interoperable with the full version. 
\new{Using our method, we thereby obtain an end-to-end symbolically verified WireGuard implementation. \newcr{All our models, code, and proofs are available online~\ifproofs{\cite{artifact}}{\cite{full-version,artifact}}}.}

% !TEX root = main.tex

%%%%%%%%%%%%%%%%%%%%%%%%%%%%%%%%%%%%%%%%%%%%%%%%%%%%%%%%%%%%
\section{Background}%% [1 page]
\label{sec:background}
%%%%%%%%%%%%%%%%%%%%%%%%%%%%%%%%%%%%%%%%%%%%%%%%%%%%%%%%%%%%

We present background on tools and methodology.

\subsection{Tamarin and multiset rewriting}
\label{ssec:tamarin-msr}
%%%%%%%%%%%%%%%%%%%%%%%%%%%%%%%%%%%%%%%%%%%%%%%%%%%%%%%%%%%%

The Tamarin prover~\cite{SchmidtMCB12,tamarin13} is an automatic,
state-of-the-art, security protocol verification tool that works in the
symbolic model.
Tamarin has been used to find 
weaknesses in and verify improvements to substantial
real-world protocols
like the 5G AKA
protocol~\cite{DBLP:conf/ccs/BasinDHRSS18,CD20195gaka}, 
TLS~1.3~\cite{CremersHHSM17},
the Noise framework~\cite{girol-noise-analysis}, and
EMV payment card protocols~\cite{DBLP:conf/sp/BasinST21,EMV-brand-mixup}.

Protocols are represented as \emph{multiset rewriting (MSR) systems}, where each
rewrite rule represents a step or action taken by a protocol
participant or the attacker.  We present the following building blocks: messages, facts, and rules.

\looseness=-1
Message \emph{terms} are elements of a \emph{term algebra} $\Terms = \TermsFull$. These are built over 
a signature $\Sigma$ of function symbols and \new{a set of names $\names = \freshtype \cup \pubtype$
consisting of} \new{a set} 
of fresh names $\freshtype$ (for secret values, 
generated by parties, unguessable by the attacker), 
a countably infinite set of public names $\pubtype$ (for globally known values), and \new{a set of} variables $\vartype$.
Cryptographic \emph{messages} $\Msgs$ are modeled as ground terms, i.e., terms without variables.
The term algebra is equipped with an \emph{equational theory} $\Eq$, which is a set of equations, and we denote by $=_\Eq$ the equality modulo $\Eq$.
\begin{example}[Diffie-Hellman equational theory]
\label{example:DH1}
The signed Diffie-Hellman (DH) protocol is a well-known key exchange protocol,
where two agents $A$ and $B$ exchange two DH public keys, $g^x$ and $g^y$, to establish the shared key $g^{xy}$ (where $g$ is a group generator).
\new{For the Tamarin model,} we use a term signature containing
symbols
$\hat{}$, $g$, $\mathit{sign}$, $\mathit{verify}$, $\mathit{pk}$ %, $\mathit{true}$
modeling respectively
exponentiation, the  group generator, signature, verification, and public keys.
We use
the simplified theory:
\[\begin{array}{r@{\qquad}l}
(g^x)^y = (g^y)^x  & \mathit{verify}(\mathit{sign}(x,k),\mathit{pk}(k)) = \mathit{true}
\end{array}
\]
Tamarin's actual model includes further equations.
The protocol's informal description is as follows: 
\[\resizebox{\linewidth}{!}{%
$\begin{array}{lll}
A \rightarrow B: & g^x                              & \text{\small $x$ fresh}\\
B \rightarrow A: & sign(\langle \new{0, B, A}, g^x,g^y\rangle,k_B) & \text{\small $y$ fresh}\\
A \rightarrow B: & sign(\langle \new{1, A, B}, g^y,g^x\rangle,k_A) & \text{\small agree on $(g^x)^y =_\Eq (g^y)^x$}\\
\end{array}
$}
\]
\new{The tags 0 and 1 are used to distinguish the last two messages.}
\end{example}

All parties, including the attacker, can use the equational theory.
The attacker also has its own set of rewriting rules, expressing
that it can
intercept, modify, block, and recombine all network messages,
following the classic Dolev-Yao (DY) model~\cite{DolevY83}.
These rules are generated automatically from the equational
theory, but users may formalize additional rules giving the attacker
further scenario-specific capabilities.

Facts are simply atomic predicates 
applied to message terms, constructed over a signature $\sigfacts=\siglin\uplus\sigper$ of \emph{fact symbols}, partitioned into \emph{linear} ($\siglin$), i.e., single-use, and \emph{persistent} ($\sigper$) facts, which encode the state of agents and the
network. 
We write $\fac=\{F(t_1,\dots,t_k)\;|\; F\in \sigfacts \text{ with arity } k, \text{ and } t_1,\dots,t_k\in\Terms\}$
for the set of facts instantiated with terms, partitioned into $\faclin\uplus\facper$ as expected.
In addition, $\cupM$, $\capM$,  $\setminusM$, $\subM$, and $\inM$ denote the usual
operations and relations on multisets, and for a multiset~$m$, $\setM(m)$ denotes the set of its elements.

A \emph{multiset rewriting (MSR) rule}, written
$\ll\rew{\aa}\rr$,
contains multisets of \emph{facts} $\ll$ and $\rr$ on
the left and right-hand side, and is labeled with $\aa$, a
multiset of actions (also facts, but disjoint from state facts) used for
property specification.
\newcr{A MSR system~$\msrsys$ is a finite set of rewrite rules.}
A MSR system~$\msrsys$ and an equational theory~$\Eq$ have a semantics as a labeled transition system (LTS), whose states are multisets of ground facts from $\fac$, the initial state is empty ($\emptyM$),
and the transition relation $\redms{\msrsys,\Eq}{\cdot}$ is defined by
\begin{equation}
\label{eq:msr-transitions}
 \inferrule
  {\ll \rew{\aa} \rr \in \msrsys \\
   \ll' \rew{\aa'} \rr' =_\Eq (\ll \rew{\aa} \rr)\theta \\
   \ll' \capM \faclin \subM S \\
   \setM(\ll') \cap \facper \subseteq \setM(S)}
  {S \redms{\msrsys,\Eq}{\aa'} S\setminusM (\ll'\capM \faclin) \cupM \rr'},
\end{equation}
where $\theta$ is a ground instance of the variables in $\ll$, $\aa$, and~$\rr$.
Intuitively, the relation describes an update of state $S$ to a successor state, that is possible when a given rule in $\msrsys$ is applicable, i.e., an instantiation with some $\theta$ \new{(mod $\Eq$)} of its left-hand side appears in $S$. Applying the rule consumes the linear \new{but not the persistent} facts appearing in its left-hand side, and adds the instantiations under $\theta$ of all the facts of its right-hand side to the resulting successor state.

\looseness=-1
The MSR rules used in Tamarin feature the reserved fact symbols $\K\in\sigper$, and $\Frf,\Inf,\Outf\in\siglin$,
modeling respectively the attacker's knowledge, freshness generation, inputs, and outputs.
The attacker is modeled by a set of \emph{message deduction rules} $\MD_\Sigma$,
giving it the DY capabilities mentioned above.
A distinguished \emph{freshness rule},
labeled $\Frf(n)$, generates fresh values $n$, which either protocol agents or the attacker directly learn, but never both. 

Finally, a protocol's observable behaviors are its \emph{traces}, which are sequences of multisets of actions labeling a sequence of transitions. We define the sets of full traces and of filtered traces with empty labels removed.
\begin{equation*}
\label{eq:msr-traces}
\begin{array}{l}
\traces(\msrsys) = \{ \tuple{a_i}_{1 \leq i \leq m} \mid  \\
\quad \exists s_1,\dots,s_m.\; \emptyM \redms{\msrsys,\Eq}{a_1} s_1 \redms{\msrsys,\Eq}{a_2}\dots\redms{\msrsys,\Eq}{a_m} s_m\}  \\[.7ex]
\tamtracesp(\msrsys) = \{ \tuple{a_i}_{1 \leq i \leq m, a_i \neq \emptyM} \mid  \tuple{a_i}_{1 \leq i \leq m} \in \traces(\msrsys) \}.
\end{array}
\end{equation*}
To ensure that fresh values are unique, we exclude traces with colliding fresh values by defining
\begin{equation*}
\begin{array}{l}
\label{eq:msr-traces-coll-free}
\tamtraces(\msrsys) = \{
\tuple{a_i}_{1 \leq i \leq m}
\in \tamtracesp(\msrsys) \mid \\
\qquad \forall i,j,n.\; \Frf(n)\in a_i\capM a_j \Rightarrow i=j
\}.
\end{array}
\end{equation*}

We will abbreviate inclusions between each kind of trace sets using the relation symbols $\tracePre$, $\tamtracepPre$, and $\tamtracePre$. For example, $\RR_1 \tamtracePre \RR_2$ denotes $\tamtraces(\RR_1) \subseteq \tamtraces(\RR_2)$, and similarly for the other two. Note that  $\tracePre \,\subseteq\, \tamtracepPre \,\subseteq\, \tamtracePre$.

We focus here on Tamarin's \emph{trace properties} (i.e., sets of traces) such as secrecy and authentication~\cite{Lowe97a}. An MSR $\RR$ satisfies a trace property $\Phi$, if $\tamtraces(\msrsys)\subseteq \Phi$. 

\begin{example}[Diffie-Hellman]
\label{example:DH2}
Continuing Example~\ref{example:DH1}, we use the linear fact symbols $\Setupf_{\Alice}(\vect{init})$,
$Step^1_{\Alice}(\vect{init}, x)$, $Step^2_{\Alice}(\vect{init}, x, g^y)$ to initialize and record 
the progress of agent $A$ playing the role of Alice in the protocol.
The parameters of the facts store her knowledge.
Alice's state is initialized with $\vect{init}=\rid, A, k_A, B, pk_B$, i.e.,
a thread identifier, her identity, her private key, and her partner's
identity and public key.
This state is then extended with her share of the secret $x$, and the DH public key 
$g^y$ she received.
For Alice, the two steps of the protocol can then be modeled by the rules:
\[\resizebox{\linewidth}{!}{%
$\begin{array}{l}
[\Setupf_{\Alice}(\vect{init}), \Frf(x)] \rew{\emptyM} [Step^1_{\Alice}(\vect{init}, x), \Outf(g^x)]
\\[1em]
[Step^1_{\Alice}(\vect{init}, x), \Inf(sign(\langle 0, B, A, g^x,Y\rangle,k_B))] \rew{[\id{Secret}(Y^x)]}\hspace{2em}~\\
\hfill [Step^2_{\Alice}(\vect{init}, x, Y), \Outf(sign(\langle 1, A, B, Y,g^x\rangle,k_A))]
\end{array}$}
\]
The received signature is checked \newcr{using pattern-matching, knowing that $pk_B = pk(k_B)$.}
The action fact $\id{Secret}(Y^x)$ in the second rule is used to specify key secrecy.
It records Alice's belief that the key she computes from the value $Y$ (supposedly $g^y$)
received from Bob remains secret.
The fact $\Setupf_{\Alice}(\vect{init})$ \new{in the first rule} is produced by 
another rule, modeling the environment \new{initializing Alice's knowledge:}
\[\resizebox{\linewidth}{!}{$%
\begin{array}{l}
   [\Frf(\rid), \skf(A,k_A), \pkf(B, pk_B)] \rew{\emptyM} %\hspace{6em} \\ 
   %\hfill 
   [\Setupf_{\Alice}(\rid, A, k_A, B, pk_B)].
\end{array}$}
\]
\end{example}

\subsection{Separation logic and I/O specifications}
\label{ssec:separation-logic-io-specs}
%%%%%%%%%%%%%%%%%%%%%%%%%%%%%%%%%%%%%%%%%%%%%%%%%%%%%%%%%%%%

% separation logic

\looseness=-1
Separation logic enables sound and modular reasoning about heap manipulating programs by associating every allocated heap location with a \emph{permission}. Permissions are a static concept used to verify programs, but do not affect their runtime behavior. Each permission is held by at most one function execution at each point in the program execution. 
\new{To access a heap location, a function must hold the associated permission;} otherwise, a verification error occurs.
The \emph{separating conjunction} $\star$ sums up the permissions in its conjuncts.
Permissions to an unbounded set of locations, for instance, all locations of a linked list, can be expressed via \new{co-recursive predicates}.
Moreover, \emph{abstract} predicates are useful to specify permissions to an unknown set of locations.

% I/O specification
\looseness=-1
Permission-based reasoning generalizes from heap locations to other kinds of program resources.
Penninckx et al.~\cite{Penninckx15} reason about a program's I/O behavior by associating each I/O operation with a permission that is required to call the operation and is then consumed.
They equip the main function's precondition with an \emph{I/O specification} that grants all permissions necessary to perform the desired I/O operations of the entire program execution.
These I/O specifications can easily be encoded in standard separation logic, such that existing program verifiers supporting different programming languages can be used to verify I/O behavior.  

Every I/O operation $\op{io}$, such as sending or receiving a value, is associated with an abstract predicate $\perm{io}$, called an \emph{I/O permission}.
Intuitively, $\perm{io}(p_1, \bar{v}, \bar{w}, p_2)$ expresses the permission to perform $\op{io}$ with outputs~$\bar{v}$ and inputs~$\bar{w}$.
We use $\bar{x}$ to denote a vector of zero or more values. 
The parameters $p_1$ and $p_2$ are called \emph{source and target places}, respectively.
An abstract predicate $\Token{p}$ is called a \emph{token} at place~$p$.
The I/O operation $\op{io}$ moves a token from the source place $p_1$ to the target place $p_2$ 
by consuming $\Token{p_1}$ and producing $\Token{p_2}$.
\new{Hence, the position of the token indicates the currently allowed I/O operations.}

%% example of an I/O operation
\begin{figure}[t]
\begin{gobra}[numbers=none]
requires token(?p$_1$) && out(p$_1$,v,?p$_2$)
ensures  ok $\Longrightarrow$ token(p$_2$)
ensures !ok $\Longrightarrow$ token(p$_1$) && out(p$_1$,v,p$_2$)
func send(v int) (ok bool)
\end{gobra}
\vspace{-0.5em}
\caption{
  Specification of the $\op{send}$~operation with I/O permissions.
  Variables starting with~\code{?} are implicitly existentially quantified.
  The code verifier uses~\code{&&} to denote the separating conjunction~$\star$.
}
\label{fig:send-io-spec}
\vspace{-.5em}
\end{figure}

\begin{example}[Send I/O operation]
Figure~\ref{fig:send-io-spec} shows the signature and specification of a \code{send} function.
The precondition requires an I/O permission~$\perm{out}$ to send the value \code{v} at some source place~$p_1$ with the corresponding token.
When the send operation succeeds, the I/O permission is consumed and the token is moved to some target place~$p_2$.
In case of failure, the token remains at the source place and the I/O permission is not consumed.
\end{example}

\looseness=-1
\new{The separating conjunction of I/O permissions with the same source place encodes non-deterministic choice between such permissions.}
Moreover, co-recursion enables repeated as well as non-terminating sequences of I/O operations.

\begin{example}[I/O specification for a server]
\[\begin{array}{r@{\;}l}
P(p,S) 		& = Q(p,S) \star R(p,S)\\
Q(p_1,S)	& = \exists v, p_2, p_3.\, \perm{in}(p_1, v, p_2) \star \perm{out}(p_2, v, p_3)\\
			& \qquad \star P(p_3, S \cup \{v\})\\
R(p_1,S)	& = \exists p_2.\, \perm{out}(p_1, \text{``Ping''}, p_2) \star P(p_2,S)
\end{array}\]
The formula $\phi = \Token{p} \star P(p,\emptyset)$ specifies a non-terminating server that repeatedly and non-deterministically chooses between \new{receiving and forwarding a value~$v$ or sending a $\text{``Ping''}$ message.
All received values~$v$ are recorded in the state~$S$, which is initially empty.}
Input parameters, like $v$ in $\perm{in}$ here, are existentially quantified to avoid imposing restrictions on the \new{values received from the environment.}
\end{example}

\looseness=-1
To enforce certain state updates between I/O operations, it is useful to associate permissions also with certain internal (that is, non-I/O) operations and include those \emph{internal  permissions} in an I/O specification. For instance, the above server could include an internal operation to reset the state $S$ when it exceeds a certain size.

\looseness=-1
I/O specifications induce a transition system and hence have a trace semantics. The traces can intuitively be seen as the sequences of I/O permissions consumed by possible executions of the programs that satisfy it.
We write $\traces(\phi)$ for the set of traces of an I/O specification~$\phi$. In the example above, the sequence $\perm{in}(5) \cdot \perm{out}(5) \cdot \perm{out}(\text{``Ping''}) \cdot \perm{in}(7) \cdot \perm{out}(7)$ is one example of a trace of $\phi$.
Note that the I/O permissions' place arguments do not appear in the trace.

% !TEX root = main.tex

%%%%%%%%%%%%%%%%%%%%%%%%%%%%%%%%%%%%%%%%%%%%%%%%%%%%%%%%%%%%
%\clearpage
\section{From Tamarin models to I/O specifications}
\label{sec:theory}
%%%%%%%%%%%%%%%%%%%%%%%%%%%%%%%%%%%%%%%%%%%%%%%%%%%%%%%%%%%%

\looseness=-1
In this section, we present our transformation of a Tamarin protocol model, given as an MSR system~$\RR$, into a set of I/O specifications $\PP_i$, one for each protocol role~$i$. The~$\psi_i$ serve as program specifications, against which the roles' implementations $c_i$ are verified (Section~\ref{sec:implementation}). Our main result is an overall soundness guarantee stating that the traces of the complete system $C(c_1, \ldots, c_n, \EE)$, composed of the roles' verified implementations $c_i$ 
and the environment $\EE$, are contained in the traces of the protocol model~$\RR$  (Section~\ref{sec:overall-soundness}): 
\[
  C(c_1, \ldots, c_n, \EE) \tamtracePre \RR.
\]
Hence, any trace property $\Phi$ proven for the protocol model, i.e., $\tamtraces(\RR) \subseteq \Phi$, is inherited by the  implementation.

The sound transformation of an MSR protocol model $\RR$ into a set of I/O specifications is challenging:%
\begin{enumerate}
\item Tamarin's MSR formalism is very general and offers great flexibility in modeling protocols and their properties. We want to preserve this generality as much as possible.

\item For the transformation to I/O specifications, we require a separate description of each protocol role and of the environment, with a clear interface between the two parts. This interface will be mapped to I/O permissions in the I/O specification and eventually to (e.g., I/O or cryptographic) library calls in the implementation.

\item The protocol models operate on abstract terms, whereas the implementation manipulates bytestrings. We need to bridge this gap in a sound manner.
\end{enumerate}

Our solution is based on a general encoding of the MSR semantics into I/O specifications. To separate the different roles' rewrite rules from each other and from the environment, we partition the fact symbols and rewrite rules accordingly. The interface between the roles and the environment is defined by identifying I/O fact symbols, for which we introduce separate I/O rules. 
This isolates the I/O operations from others and allows us to map them to I/O permissions and later to library functions. Moreover, we keep I/O specifications as abstract as possible by using message terms rather than bytestrings.
We handle the transition to bytestrings in the code verification process (Section~\ref{sec:implementation}).

The proofs for the results stated in this section can be found in \ifproofs{Appendix~\ref{app:proofs}.}{the full version of this paper~\cite{full-version}.}

\subsection{Protocol format}
\label{sec:formatting}
%%%%%%%%%%%%%%%%%%%%%%%%%%%%%%%%%%%%%%%%%%%%%%%%%%%%%%%%%%%%

We introduce a few mild formatting assumptions on the Tamarin model.
They mostly correspond to common modeling practice and serve to cleanly separate the 
different protocol roles and the environment. They do not restrict
Tamarin's expressiveness for modeling protocols.
\new{In particular, all protocols in the Tamarin distribution would fit our assumptions after some minor modifications.}

\subsubsection{Rule format}
\label{ssec:rule-format}

To model an $n$-role protocol, we use a fact signature of the form
\[
\sigfacts = \sigact \udis \sigenv \udis \Big(\bigudis_{1\leq i \leq n} \sigstate{i}\Big),
\]
%\vspace{-0.5em}
\noindent where \new{(i)} $\sigact$, $\sigenv$, and $\sigstate{i}$ are mutually disjoint sets of fact symbols, used to construct action facts (used in transition labels), environment facts, and each role $i$'s state facts;
\new{(ii)} $\sigenv$ contains two disjoint subsets, $\sigin$ and $\sigout$, of input and output fact symbols \new{such that} $\Frf, \Inf \in \sigin$, $\Outf \in \sigout$, and $\K \in \sigenv \setminus (\sigin \cup \sigout)$; \new{and (iii)} there is an initialization fact symbol $\Setupf_i \in \sigin$ for each protocol role $i$. 

 We consider MSR systems $\msrsys$ whose rules are given by:
\[
\msrsys = \Renv \udis \Big(\bigudis_{1 \leq i \leq n} \msrsys_i\Big).
\]

\vspace{-0.5em}
\noindent Here, $\Renv$ and the $\msrsys_i$s
are pairwise disjoint rule sets containing
rules for the environment and each protocol role.
Protocol rules use input and output facts to communicate with the environment. For example, the following two environment rules transfer a message to and from the attacker's knowledge:

\vspace{-1em}
\begin{align}
   [\Outf(x)] \rew{\emptyM} [\K(x)]  \qquad\quad   [\K(x)] \rew{[]} [\Inf(x)].      
\end{align}
The
attacker rules $\MD_\Sigma$, the freshness rule, and the rules that generate the $\Setupf_i$ facts (cf.~Example~\ref{example:DH2}), are also in $\Renv$.
The protocol rules for role $i$ (and only those) use role state facts from $\sigstate{i}$ to keep track of the role's progress. We require that their first $k_i \geq 1$ arguments are reserved for role $i$'s parameters from the $\Setupf_i$ fact, and that the first of these parameters is the thread identifier $\rid$. 

For a more detailed specification of the format restrictions on the rewrite rules, see Appendix~\ref{app:protocol-format-details}.

\subsubsection{Protocol messages}
\label{ssec:protocol-messages}

\new{We support both the usual ways of checking received messages using pattern matching and explicit equality checks. The latter are formalized, as usual in Tamarin, as a combination of action facts labeling the given rule (e.g., $\mathsf{Eq}(x, \hash{z})$) and restrictions associating these facts with (a boolean combination of) equalities (e.g., $x =_\Eq y$  whenever $\mathsf{Eq}(x,y)$ occurs in a trace). These restrictions act as assumptions on the traces considered by Tamarin.%
\footnote{\new{To handle these, we use a modified, but equivalent MSR semantics, where the equalities are checked at each step rather than globally on traces. This semantics allows us to translate these checks into the I/O specifications and thus enforce their correct implementation. For simplicity, we present our results under the standard semantics and refer to~\newcr{\ifproofs{Appendix~\ref{app:proofs}}{\cite{full-version}}} for more details.}}
The I/O specifications resulting from our procedure require that these equality checks are implemented (Section~\ref{ssec:io-specs}).}

\new{Furthermore, we recommend, but do not require, the replacement of nested pairs and tuples by \emph{formats}~\cite{DBLP:conf/csfw/ModersheimK14}. These are user-defined function symbols, along with projections for all arguments, that behave like tuples. In the implementation, each format is mapped to a combination of tags (i.e., constant bytestrings), fixed-size fields, and variable-sized fields prepended with a length field. Formats help to soundly relate term and bytestring messages, if we prove that they are unambiguous and non-overlapping (see Section~\ref{sec:implementation}).}

\begin{example}[Diffie-Hellman formatting]
\label{example:DH4}
The rules for Alice's role from Example~\ref{example:DH2} satisfy the format conditions above.
The initiator setup rule produces a fact $\Setupf_{\Alice}(\vect{init})$, whose parameters $\vect{init}$ appear as the first parameters 
of the state facts $Step^1_{\Alice}(\vect{init}, \ldots)$ and $Step^2_{\Alice}(\vect{init}, \ldots)$. Both protocol rules produce an $\Outf$ fact to send a message. The second rule also consumes an $\Inf$ fact to receive a message.
\new{To follow our recommendation to use formats, we can model the message $\langle 0, B, A, g^x, Y\rangle$ being signed as a format with five fields, rather than a tuple.
Note that, at the Tamarin level, tuples containing unique tags to distinguish them behave essentially the same as formats.}
\end{example}

\subsection{Transformation to component models}
%%%%%%%%%%%%%%%%%%%%%%%%%%%%%%%%%%%%%%%%%%%%%%%%%%%%%%%%%%%%

\looseness=-1
We decompose an MSR system $\msrsys$ that satisfies our format requirements
into several component models, one for each role, and 
a separate environment model, which includes the attacker. 
In doing so, we move from a global view of the protocol,  useful 
for security analysis,
to a local view of each role, more appropriate 
for their implementation.
In Section~\ref{ssec:io-specs}, we transform the component models into 
I/O specifications for the programs implementing them.

As a preparatory step, we refine $\msrsys$ into
an \emph{interface model} which starts decoupling the roles from the environment 
by introducing separate rewrite rules for their interactions.

\subsubsection{Interface model}
\label{ssec:interface-model}
%%%%%%%%%%%%%%%%%%%%%

The protocol roles and the environment interact using input and output facts, including the built-in facts 
$\Frf$,  $\Inf$, and $\Outf$. For example, the protocol roles receive messages by consuming 
$\Inf$ facts produced by the attacker. The interface model adds an \emph{I/O rule} for each such fact, 
which turns it into a buffered version.
These I/O rules will later be implemented as calls to library 
functions.

\looseness=-1
\new{Let $\sigin^-$ be the set $\sigin$ without the initialization facts $\Setupf_i$.}
We first add to the fact signature, for each non-setup input or output fact $F$ and role $i$, 
a copy (the ``buffer'') $F_{i}$: 
\[\resizebox{\linewidth}{!}{%
$\begin{array}{c}
\sigbuf{i} = \{F_{i} \mid F \in \new{\sigin^-} \cup \sigout \} \quad\quad
\sigrole{i}  = \sigstate{i} \cup \sigbuf{i}\\
[0.4em]
\sigfacts'   = \sigact \udis \sigenv \udis (\bigudis_{i} \sigrole{i}).
\end{array}$}\]

We then replace the facts used by the protocol rules as follows.
For each role $i$, let $\msrsys'_{i}$ be the set of rules obtained by replacing,
in all rules in $\msrsys_{i}$, each fact $F(t_1, \ldots, t_k)$ such that $F \in \sigin^- \cup \sigout$ 
by $F_{i}(\rid, t_1, \ldots, t_k)$. The latter fact has $\rid$ as an additional parameter.

We also introduce the set $\Rbuf$ of \emph{I/O rules}, which translate between input or output 
facts and their buffered versions. The set $\Rbuf$ contains the following rules, for each role $i$. 
\begin{align*}
&[F(x_1,\dots,x_k)] \rew{\new{\emptyM}}  [F_{i}(\rid, x_1,\dots,x_k)] && \text{for $F \in \sigin^-$} \\
&[G_{i}(\rid, x_1,\dots,x_k)] \rew{\new{\emptyM}}  [G(x_1,\dots,x_k)] && \text{for $G \in \sigout$}.
\end{align*}

\vspace{-0.2em}

For reasons that will become clear later, we also count the role setup rules as I/O rules. Hence, we move them from $\Renv$ 
to $\Rbuf$, calling the remaining environment rules $\Renv^-$.

Finally, the interface model is specified by:
\begin{equation}
\label{eq:interface-model-rules}
\Rintf = \Renv^- \udis \Rbuf \udis (\bigudis_{i} \msrsys'_{i}).
\end{equation}

\vspace{-1em}
\begin{example}
\label{example:DH5}
Continuing Example~\ref{example:DH4}, we introduce the buffer facts $\Inf_{\Alice}$, $\Outf_{\Alice}$, and
$\Frf_{\Alice}$.
Recall that $\rid$ is included in $\vect{init}$.
This yields the modified set $\RR'_{\Alice}$ for the role $\Alice$:
\vspace{-0.5em}
\[\resizebox{\linewidth}{!}{%
$\begin{array}{l}
[\Setupf_{\Alice}(\vect{init}), \Frf_{\Alice}(\rid,x)] \rew{\emptyM}\\\hfill [Step^1_{\Alice}(\vect{init}, x), \Outf_{\Alice}(\rid, g^x)]\\
[0.5em]
[Step^1_{\Alice}(\vect{init}, x), \Inf_{\Alice}(\rid, sign(\langle 0,B,A,g^x,Y\rangle,k_B))] \rew{[\id{Secret}(Y^x)]}\\\hfill [Step^2_{\Alice}(\vect{init}, x, Y), \Outf_{\Alice}(\rid, sign(\langle 1,A,B,Y,g^x\rangle,k_A))].
\end{array}$
}\]
\end{example}
\looseness=-1
We show that the interface model refines the original one.
\begin{lemma}
\label{lem:interface-model}
$\Rintf \tamtracepPre \msrsys$.
\end{lemma}

\subsubsection{Decomposition}
\label{ssec:decomposition}
%%%%%%%%%%%%%%%%%%%%%%%%%%%%%%

\looseness=-1
We can now decompose the interface model into the role components and the environment.
In a nutshell, we assign the rules $\msrsys'_{i}$ to the component for role~$i$ 
and the rules $\Renv^-$, including the attacker rules $\MD_\Sigma$, to the environment. 
The protocol communicates with the environment using the I/O rules. We split them 
into two synchronized parts, one belonging to the environment and the other to the protocol. 
Below, we show that the re-composed system implements the interface model.

\looseness=-1
We explain the splitting of the I/O rules into a protocol part and an environment part using 
the example rule
\[
[\Outf_{i}(\rid, x)] \rew{\emptyM} [\Outf(x)]. 
\]
This rule models instance $\rid$ of role $i$ outputting a message to the attacker.
We split this rule into two parts: 
\begin{align}
[\Outf_{i}(\rid, x)] & \rew{[\lambda_{\Outf}(\rid, x)]} \emptyM,  \label{eq:out-proto} \\
\emptyM                         & \rew{[\lambda_{\Outf}(\rid, x)]} [\Outf(x)],  \label{eq:out-env}
\end{align}
where the first rule belongs to role $i$ and the second to the environment.
We label both rules with a new action fact $\lambda_{\Outf}(\rid, x)$, which uniquely 
identifies the original I/O rule and has as parameters all variables occurring in it. 
We call this fact a \emph{synchronization label}, as we later use it
for synchronizing the two parts to recover the original rule's behavior.
We similarly split all rules in $\Rbuf$, yielding two sets $\Rbuf^{i}$ and $\Rbuf^e$ 
belonging to the protocol role $i$ and to the environment.

The components for each protocol role $i$ and for the environment are then 
defined as follows.
\begin{align}
\Rrole{i} & = \RR'_{i} \udis \Rbuf^{i} &
\Renv^e & = \Renv^- \udis \Rbuf^{e}.
\end{align}
Note that these two rule sets operate on pairwise disjoint sets of facts, 
namely over the signatures $\sigrole{i}$ and $\sigenv$, respectively. 
Hence they can interact with each other only by synchronizing the split I/O rules.

\begin{example}[Component for Diffie-Hellman]
\label{example:DH7}
The MSR system $\Rrole{\Alice}$ for Alice's role contains the two protocol rules in $\RR'_{\Alice}$
from Example~\ref{example:DH5}, the output rule~\eqref{eq:out-proto} described above, 
and similar rules for inputs and freshness generation. In addition, it contains the protocol part of the split 
setup rule for Alice's role from~Example~\ref{example:DH2}, i.e.
\vspace{-0.5em}
\[
   \emptyM \rew{[\lambda_{\Alice}(\rid, A, k_A, B, pk_B)]} [\Setupf_{\Alice}(\rid, A, k_A, B, pk_B)].
\]
\end{example}
%\vspace{-0.5em}

\looseness=-1
The traces of the recomposition of all roles with the environment
are included in the traces of the interface model.
We define two kinds of parallel compositions on LTSs 
(induced here by the MSR systems' transition semantics). 
\new{We give the intuition here and refer the reader to Appendix~\ref{app:composition-defs} for the formal definitions.}
The (indexed) parallel composition $|||$ interleaves the transitions of a family of component systems without communication. 
The (binary) parallel composition $\sync{\Lambda}$ synchronizes transitions with labels from the set $\Lambda$, resulting in a transition labeled~$\emptyM$, and interleaves all other transitions. 
We then show:
\begin{lemma}[Decomposition]
\label{lem:decomposition}
Let $\Rrole{i}(\rid)$ be the MSR system $\Rrole{i}$ for a fixed thread id $\rid$. Then 
\[
  (|||_{i,\rid}\Rrole{i}(\rid)) \sync{\Lambda} \Renv^e \tracePre \Rintf,
\]
where $\Lambda = \bigcup_i \{ \aa\theta \mid \exists \ll, \rr.\,\ll \rew{\aa} \rr \in \Rbuf^i \land \id{range}(\theta) \subseteq \Msgs \}$ consists of all ground instances of 
synchronization labels.
\end{lemma}

\subsection{Transformation to I/O specifications}
\label{ssec:io-specs}
%%%%%%%%%%%%%%%%%%%%%%%%%%%%%%%%%%%%%%%%%%%%%%%%%%%%%%%%%%%%

\looseness=-1
Finally, we extract an I/O specification $\PP_{i}$ from each role $i$'s MSR system $\Rrole{i}$, which serves as the specification for the role's implementation at the code level.
$\PP_i$ is parameterized by the thread identifier $\rid$, and it associates a token with the starting place $p$ of the predicate $\IOSpecRole{i}$:
\[
\PP_{i}(\rid) = \exists p.\: \id{token}(p) \star \IOSpecRole{i}(p,\rid,\emptyM). 
\]

The predicate $\IOSpecRole{i}(p,\rid,S)$'s parameters are a place~$p$, a thread identifier~$\rid$, and a state~$S$ of the MSR system $\Rrole{i}$ (i.e., a multiset of ground facts). Note that $\PP_{i}$ invokes $\IOSpecRole{i}$ with the initial state, i.e., the empty multiset~$\emptyM$ (see Section~\ref{ssec:tamarin-msr}). It is defined co-recursively as the separating conjunction over the formulas $\phi_{R}$, one for each rewrite rule $R \in \Rrole{i}$:
\[
\IOSpecRole{i}(p,\rid,S) = \bigstar_{R \in \Rrole{i}} \: \phi_R(p,\rid,S).
\]

$\phi_{R}$ encodes an application of the rewrite rule $R$ to the model state $S$. It contains an I/O or internal permission~$\perm{R}(p, \ldots, p')$, which an implementation must hold in order to execute the program part implementing $R$.
$\phi_R$ co-recursively calls $\IOSpecRole{i}(p',\rid,S')$ with the target place $p'$ of $\perm{R}$ and the updated state $S'$. We define the formulas $\phi_{R}$ separately for protocol rules in $\RR'_i$ and for I/O rules in $\Rbuf^i$. 

Consider a protocol rule $R = \ll \rew{\aa} \rr \in \RR'_{i}$ with variables~$\vect{x}$.
We associate the internal permission $\perm{R}(p,\vect{x},\ll',\aa',\rr',p')$ to $R$, and define $\phi_R(p,\rid,S)$ by
\[\resizebox{\linewidth}{!}{%
$\begin{array}{l}
\phi_{R}(p,\rid,S) = \forall \vect{x}, \ll', \aa', \rr'.\;\\
\ \ M(\ll',s) \,\land\, \ll' =_\Eq \ll \,\land\, \aa'=_\Eq\aa \:\land\: \rr'=_\Eq\rr  \new{\:\land\: \Phi_R(\vect{x})} \\
\ \ \Longrightarrow
\exists p'.\; \perm{R}(p,\vect{x},\ll',\aa',\rr',p') \,\star\, \IOSpecRole{i}(p',\rid,U(\ll',\rr',S))\\
\end{array}$
}
\]
where 
$M(\ll', S)      = (\ll' \capM \faclin \subM S) \land (\setM(\ll') \cap \facper \subseteq \setM(S))$,  
$U(\ll', \rr', S) = S \setminusM ( \ll' \capM \faclin) \cupM \rr'$, 
\new{and
$\Phi_R$ is the conjunction of all (boolean combinations of) equality checks \new{(mod $\Eq$)} that the rule $R$ performs using a combination of action facts in $\aa$ and associated restrictions (cf.~Section~\ref{ssec:protocol-messages}).}

This formula specifies that, for any instantiation (mod~$\Eq$) $\ll', \aa', \rr'$ of the facts in the rule, 
if the matching condition $M(\ll',S)$ \new{and the equational formula $\Phi_R$ are} satisfied, then we have an internal permission $\perm{R}$ to execute the rule's 
implementation. 
This yields an updated state $U(\ll',\rr',S)$, on which $\IOSpecRole{i}$ is co-recursively 
applied to produce the permissions for the rest of the execution.
\new{The formula $\Phi_R$ thus enforces that the implementation performs the explicit equality checks on messages specified in the rule $R$ (see also Section~\ref{ssec:proving-equalities}).}

We define similar formulas for all I/O rules. For output rules of the form $R_G = [G_i(\rid, \vect{x})] \rew{[\lambda_{G}(\rid, \vect{x})]} \emptyM \in \Rbuf^i$, we define the formula
\[\resizebox{\linewidth}{!}{%
$\begin{array}{l}
\phi_{R_G}(p,\rid,S) =\forall \vect{x}.\;
G_i(\rid, \vect{x}) \inM S \\
\ \ \Longrightarrow 
\exists p'.\, \perm{R_G}(p, \rid, \vect{x}, p') \star \IOSpecRole{i}(p', \rid, S \setminusM [G_i(\rid, \vect{x})])), \\
\end{array}$}
\]
which grants the I/O permission $\perm{R_G}(p, \rid, \vect{x}, p')$ for any $\rid$ and terms $\vect{x}$ for which the fact $G_i(\rid, \vect{x})$ exists in $S$. $\IOSpecRole{i}$ is called co-recursively with the target place of the permission and the updated state, where $G_i(\rid, \vect{x})$ is removed.

For input rules $R_F = \emptyM \rew{[\lambda_{F}(\rid, \vect{z})]} [F_i(\rid, \vect{z})] \in \Rbuf^i$, we define the formula
\[
\begin{array}{l}
\phi_{R_F}(p,\rid,S) = \\
\ \ \exists \new{p'}, \vect{z}.\; \perm{R_F}(p,\rid, \vect{z}, p') \,\star\, \IOSpecRole{i}(p', \rid, S \cupM [F_i(\rid, \vect{z})])),
\end{array}
\]
which grants the I/O permission $\perm{R_F}(p, \rid, \vect{z}, p')$ to read inputs~$\vect{z}$ for any $\rid$. Note that the input variables $\vect{z}$ are existentially quantified (cf.~Section~\ref{ssec:separation-logic-io-specs}) and the fact $F_i(\rid, \vect{z})$ is added to the state in the co-recursive call to $\IOSpecRole{i}$.

\begin{example}[I/O specification for Diffie-Hellman]
\label{example:DH9}
Continuing Examples~\ref{example:DH5} and~\ref{example:DH7}, the component system
is translated into an I/O specification that features
the following conjunct corresponding to the rule for the second step of Alice's
role:
\[\resizebox{\linewidth}{!}{$%
\begin{array}{l}
\phi_{\Alice_2}(p,\rid,S) = \forall \vect{init}, x, Y, \ll', \aa', \rr'.\;\\
  \;\; M(\ll',S) \wedge {} \\
  \;\; \ll'=_\Eq\multileft Step^1_{\Alice}(\vect{init}, x), \Inf_{\Alice}(\rid,sign(\langle 0,B,A,g^x,Y\rangle,k_B))\multiright \wedge {} \\
  \;\; \aa'=_\Eq\multileft \id{Secret}(Y^x)\multiright \wedge {}\\
  \;\; \rr'=_\Eq\multileft Step^2_{\Alice}(\vect{init}, x, Y), \Outf_{\Alice}(\rid, sign(\langle 1,A,B,Y,g^x\rangle,k_A))\multiright \wedge {} \\
  \;\; \Longrightarrow
  \exists p'.\, \perm{Alice_2}(p,\vect{init},x,Y,\ll',\aa',\rr',p') \!\star\! \IOSpecRole{i}(p',\rid,U(\ll',\rr',S)).\\
\end{array}
$}\]
Namely, the permission to execute this step is granted, provided
$S$ contains instantiations of the previous state fact
and the correct input fact, and that $S$ is updated by replacing them
with the new state and output facts. \new{Recall that $\vect{init}$ abbreviates $\rid, A, k_A, B, pk_B$.}
\end{example}

\new{The construction of $\PP_i$ from $\Rrole{i}$ can be seen as an instance of
the method presented in~\cite{Igloo} and by the soundness 
result from that paper, we get the following trace inclusion.}

\begin{theorem}[\cite{Igloo}]
\label{theorem:IO-soundness}
For all MSR systems $\Rrole{i}(rid)$, where the fresh name $\rid$
instantiates the
thread id in all facts,
\[
\med(\PP_{i}(\rid)) \tracePre \Rrole{i}(\rid),
\]
where $\med = \medint \circ \medext$ relabels (the LTS induced by) $\PP_{i}(\rid)$. Here, $\medint$ and $\medext$ are the identity functions except on the following labels:
\begin{align*}
\medint(\perm{R}(\vect{x},\ll',\aa',\rr')) & = \aa' 
&& \text{for $R \in \RR'_i$} \\
\medext(\perm{F}(\rid, \vect{x}))          & = \enumM{\lambda_F(\rid, \vect{x})}
&& \text{for $F \in \Rbuf^i$.}
\end{align*}
\end{theorem}

% !TEX root = main.tex

%%%%%%%%%%%%%%%%%%%%%%%%%%%%%%%%%%%%%%%%%%%%%%%%%%%%%%%%%%%%
\section{Verified protocol implementations}
\label{sec:implementation}
%%%%%%%%%%%%%%%%%%%%%%%%%%%%%%%%%%%%%%%%%%%%%%%%%%%%%%%%%%%%

The implementation step consists of providing the code $\prog_i(rid)$ implementing each role $i$ and proving that it satisfies its I/O specification $\PP_{i}(\rid)$. The challenge here is bridging the abstraction gap between the message terms in the I/O specifications $\PP_{i}(\rid)$ and the bytestring messages manipulated by the code. 
In Section~\ref{ssec:code-verif-with-abstraction}, we present an extension of the code verifiers' \new{semantics of Hoare triples} to accommodate such abstractions. In~Section~\ref{ssec:terms-to-bytes}, we explain how we concretely relate bytestrings to terms. Finally,  in Section~\ref{ssec:library-specs}, we show how we verify the roles' I/O specifications based on appropriate I/O and cryptographic library specifications.

\subsection{Code verification with abstraction}
\label{ssec:code-verif-with-abstraction}
%%%%%%%%%%%%%%%%%%%%%%%%%%%%%%%%%%%%%%%%%%%%%%%%%%%%%%%%%%%%

In Penninckx et al.'s program logic~\cite{Penninckx15}, the statement that a program $c$ satisfies an I/O specification $\phi$ is expressed as the Hoare triple
\begin{align}
  \label{eq:prog-spec-basic}
  \{ \phi \}\ \prog \ \{ \kw{true} \},
\end{align}
with the I/O specification $\phi$ in the precondition and the post-condition $\kw{true}$. We assume that the program $\prog$ has an LTS semantics $\CC$ given by the programming language's operational semantics, where the labels represent the program's I/O (and internal) operations and the program's traces consist of sequences of such labels. We leave the exact semantics unspecified here, to keep our formulation generic with respect to the programming language used. The semantics of the Hoare triple~\eqref{eq:prog-spec-basic} implies that the program $\prog$'s traces are included in the traces of $\phi$, i.e. $\CC \tracePre \phi$. 

\new{To bridge the gap between message terms and bytestring messages}, we extend Penninckx et al.'s approach by introducing an \emph{abstraction} or relabeling function $\alpha$ between the implementation's transition labels and the I/O specification's transition labels. For example, $\alpha$ may map a concrete label $\perm{in}_c(l)$ to an abstract version $\perm{in}_a(s)$, where $l$ is a list implementation and $s$ is the mathematical set of $l$'s elements. We also extend the soundness assumption on the code verifier accordingly.
\begin{newtext}
\begin{assumption}[Verifier assumption]
\label{assm:verifier-assumption}
\begin{align*}
  \vdash_\alpha \{ \phi \}\ \prog \ \{ \mathsf{true} \} \implies \alpha(\CC) \tracePre \phi.
\end{align*}
\end{assumption}
\end{newtext}
This means that a successful verification implies that the program traces, abstracted under $\alpha$, are included in the I/O specification $\phi$'s traces. \new{In Appendix~\ref{app:verifier-assumption}, we sketch a semantics for Hoare triples
that entails this trace inclusion.}

To ensure this extension's soundness, we 
require that the \new{I/O 
operations' contracts are \emph{consistent with $\alpha$}, i.e.,} they imply a correct mapping of transition labels under $\alpha$. 
More precisely, suppose the specification of such an operation \code{op} induces a concrete transition label $\perm{op}_c(\vect{a})$, where $\vect{a}$ are \code{op}'s inputs and outputs, and the I/O permission in the precondition induces the abstract transition label $\perm{op}_a(\vect{b})$. Then we define $\alpha$ as lifting from an (overloaded) function $\alpha$ that maps concrete parameter types to abstract ones, i.e., $\alpha(\perm{op}_c(\vect{a})) = \perm{op}_a(\alpha(\vect{a}))$. We therefore require that $\vect{b} = \alpha(\vect{a})$ follows from \code{op}'s precondition (for arguments) and postcondition (for return values). 
Moreover, we allow $\alpha$ to be a partial function, in which case the specification must also imply that the concrete arguments are in its domain.

\subsubsection*{Application to role verification}
We now apply this idea to the verification of the protocol's role implementations (using Gobra in our case study). That is, we wish to establish
\begin{align}
  \label{eq:prog-spec}
  \vdash_{\alpha} \{\PP_{i}(\rid) \}\ \prog_i(rid) \ \{ \mathsf{true} \}
\end{align}
for a suitable $\alpha$. 
An obvious possibility would be to define an abstraction function $\alpha\!: \bytestrings \fun \Msgs$ from bytestrings to messages and then lift it to trace labels. For example, a concrete $\perm{in}_c(\rid, b)$ would be abstracted to $\alpha(\perm{in}_c(\rid, b)) = \perm{in}_a(\rid, \alpha(b))$. However, this mapping assumes that each bytestring corresponds to exactly one term, and consequently that  \emph{every} bytestring can be uniquely parsed as a term. To minimize our assumptions, however, we do not a priori want to exclude collisions between bytestrings, i.e., we allow a bytestring to have several term interpretations. 

In Section~\ref{ssec:terms-to-bytes}, we therefore relate bytestrings and terms using a concretization function $\gamma\!:  \Msgs \fun \bytestrings$. Since a bytestring may be related to several terms, we cannot define a function $\alpha$ mapping concrete labels to abstract I/O labels. Our solution is based on adding ghost term parameters to the I/O operations in the implementation code.  For example, the operation receiving a bytestring $b$ gets an additional ghost return value term \new{$m$} with $b=\gamma(\new{m})$ and the corresponding transition  label is $\perm{in}_c(\rid, (b,\new{m}))$.  
These ghost terms \new{aid} verification (see Section~\ref{ssec:library-specs}), but are not present in the executable code.
We instantiate $\alpha$ to the function~$\medextp$ that removes the bytestrings from the concrete I/O operation's labels and keeps only the ghost terms used for the reasoning. For instance, $\medextp(\perm{in}_c(\rid, (b,\new{m}))) = \perm{in}_a(\rid, \new{m})$. This function is defined only for $b = \gamma(\new{m})$, which is guaranteed by the receive operation's contract (cf. Figure~\ref{fig:netlib-contracts}).

Our proposed method enables us to verify that pre-existing real-world code satisfies I/O specifications produced from abstract Tamarin models (see Section~\ref{sec:wireguard}).

\subsection{Relating terms and bytestrings}
\label{ssec:terms-to-bytes}
%%%%%%%%%%%%%%%%%%%%%%%%%%%%%%%%%%%%%%%%%%%%%%%%%%%%%%%%%%%%

In Tamarin's MSR semantics, messages in $\Msgs$ are ground terms.
We model the concrete messages and the operations on them as \emph{\new{bytestring algebras}} defined as $\Sigma$-algebras $\Calg$ with the set of bytestrings $\bytestrings$
as the carrier set. 
To relate terms to bytestrings, we use a surjective $\Sigma$-algebra homomorphism 
$
  \gamma \! : \Msgs \fun \Calg, 
$
which maps \new{(fresh and public) names} to bytestrings and the signature's symbols to functions on bytestrings:
\[\begin{array}{rcll}
   \gamma(n)\!\!&=&\!\!\Calgify{n} 
   & \text{\new{for $n \in \names$}} \\
   \gamma(f(t_1, \ldots, t_k))\!\!&=&\!\!\Calgify{f}(\gamma(t_1), \ldots, \gamma(t_k))
   & \text{for $f \in \Sigma^k$}
\end{array}\]
With the requirement that $\gamma$ is surjective, we avoid junk bytestrings that do not represent any term (i.e., the algebra $\Calg$ is term-generated). \new{This is without loss of generality as there are countably infinitely many public names that can be mapped to potential junk bytestrings.}

Note that $\Sigma$-algebra homomorphisms are required to preserve equalities. For example, a symbolic equality $\dec{k}{\enc{k}{m}} =_\Eq m$ on terms implies the equality
\begin{align*}
\decB{\id{key}}{\encB{\id{key}}{\id{msg}}} = \id{msg}
\end{align*}
on bytestrings. 
In what follows, we will use the bytestring algebra's functions in our cryptographic library's specification. This enables us to reason about message parsing and construction.

\subsection{Verifying the I/O specification}
\label{ssec:library-specs}
%%%%%%%%%%%%%%%%%%%%%%%%%%%%%%%%%%%%%%%%%%%%%%%%%%%%%%%%%%%%

\begin{figure}[t]
  \makeatletter
  \lst@AddToHook{OnEmptyLine}{\vspace{-0.4\baselineskip}}
  \makeatother
\begin{gobra}[numbers=none]
ensures seq(ciph) $=$ enc$^B$(seq(key), seq(msg))
func encrypt(key, msg []byte) (ciph []byte)

ensures ok $\Longrightarrow$ seq(c) $=$ enc$^B$(seq(k), seq(m))
func decrypt(k, c []byte) (m []byte, ok bool)

requires token(?p$_1$) && $\texttt{in}$(p$_1$,?m,?p$_2$)
ensures  ok $\Longrightarrow$ token(p$_2$) && seq(b) $=$ $\gamma$(m)
ensures !ok $\Longrightarrow$ token(p$_1$) && $\texttt{in}$(p$_1$,m,p$_2$)
func receive() (b []byte,ghost m term,ok bool)
\end{gobra}
%\vspace{-0.7em}
\caption{Simplified specifications for encryption, decryption, and receive. The function \code|seq| abstracts an in-memory byte array to $\Calg$. We omit Gobra's memory annotations needed to reason about heap data structures and conditions on the size of bytestrings.}
\label{fig:cryptolib-contracts}
\label{fig:netlib-contracts}
\vspace{-0.5em}
\end{figure}

The verification of the I/O specification generally follows the same approach as in previous work~\cite{Penninckx15,Igloo}. 
Every I/O operation performed by the code requires that a corresponding I/O permission is held. The required I/O permissions must be obtained from the I/O specification. 
However, our introduction of abstraction makes reasoning about what is sent and, in particular, received
more challenging.

\subsubsection{Sending and receiving messages}
For a sent pair of a bytestring and a ghost message (in $\Msgs$), we must verify that the I/O specification permits sending the message.
Similarly, for a received pair of a bytestring and a ghost message, we must verify that the received message matches a term in the I/O specification, describing the expected protocol message.
We refer to such terms as \emph{patterns}. 
In Example~\ref{example:DH9}, there is a single pattern, namely $sign(\langle \new{0,B,A},g^x,Y\rangle,k_B)$, 
where \new{the unconstrained~$Y$ is a variable and all other entities are constrained by the fact $Step^1_{\Alice}(\vect{init}, x)$.} 

Verifying that the I/O specification permits sending a message boils down to verifying that \new{the} bytestring $\gamma(m)$ for a permitted message $m$ was constructed and then sent.
This becomes straightforward by equipping the cryptographic library with suitable specifications.
Consider the simplified specification of an encryption function shown in Figure~\ref{fig:cryptolib-contracts}.
The function \code{seq} abstracts an in-memory byte array into a mathematical sequence of bytes, i.e., an element of~$\bytestrings$.
Due to the specification and the surjectivity of $\gamma$, 
the result of \code{encrypt(key, msg)} 
is equal to $\gamma(\enc{\Msg_{\id{key}}}{\Msg_{\id{msg}}})$ for some messages $\Msg_{\id{key}}$ and $\Msg_{\id{msg}}$,
where $\gamma(\Msg_{\id{key}}) = \texttt{seq(key)}$ and $\gamma(\Msg_{\id{msg}}) = \texttt{seq(msg)}$.
To verify the construction of an entire message, we combine the information of all such calls. 

Verifying that a message~$\Msg$ \new{returned by \code{receive()} (cf. Figure~\ref{fig:netlib-contracts})} matches a pattern~$t$ is more involved.
Using our cryptographic library's specifications, we can verify that \new{$\gamma(\Msg)$} is equal to $\gamma(t\sigma)$, where
the substitution $\sigma$ instantiates the variables of $t$ with messages.
Unfortunately, this does not entail that the received message $\Msg$ matches the pattern~$t$.
The function $\gamma$ may have collisions and hence $\gamma(\Msg)$ may \new{equal} $\gamma(t\sigma)$, while $\Msg$ and $t\sigma$ differ.
We address this issue by requiring that \new{instances of} the I/O specification's patterns do not \new{collide with other bytestrings};
we discuss below how we justify this requirement.

\begin{newtext}
\begin{definition}
The \emph{pattern requirement} for a pattern $t \in \Terms$ is defined by
\begin{align}
\label{eq:qaxiom}\tag{PaR($t$)}
     \gamma(t\sigma) = \gamma(\Msg) \implies \exists \sigma'.\: \Msg =_\Eq t\sigma'.
\end{align}
\end{definition}
\end{newtext}
\new{This requirement} states that if messages $\Msg$ and $t\sigma$ \new{coincide} under $\gamma$, then $\Msg$ must match the pattern $t$ \new{(mod $\Eq$)} with some substitution $\sigma'$, \new{which may differ from $\sigma$}.

\looseness=-1
We \new{need} the pattern requirement for all patterns of the I/O specification.
For code verification, we express \new{the pattern requirement as} a ghost function whose pre- and postcondition are the left-hand and right-hand side of the pattern requirement, \new{for each pattern respectively}.
To apply the pattern requirement, the corresponding ghost function is called in the code.
\new{In Sections~\ref{ssec:deriving-pattern-requirement} and~\ref{ssec:assumption-proof-obligations}, we will explain how to prove the pattern requirement for a given pattern $t$.}

\begin{figure}[t]
\begin{gobra}[numbers=left,numbersep=5pt]
// seq(key) $=$ $\gamma$(k) holds
ciph, c, ok := receive(); if !ok {return}
assert seq(ciph) $=$ $\gamma$(c) 
msg, ok := decrypt(key, ciph); if !ok {return}
assert $\exists$u. seq(msg) $=$ $\gamma$(u)
        && seq(ciph) $=$ $\gamma$(enc(k, u))
PaR1(m, ...) // using the pattern requirement
assert $\exists$w. c $=_\Eq$ enc(k, w) && seq(msg) $=$ $\gamma$(w) 
\end{gobra}
%\vspace{-0.7em}
\caption{Reasoning about receiving and parsing a ciphertext.}
\label{fig:parsing-ciphertext}
%\vspace{-0.5em}
\end{figure}

\begin{figure}[t]
\begin{gobra}[numbers=none]
req token(p)$\;\,$&&$\;\,$$\IOSpecRole{\id{Ann}}$(p,r,S)$\;\,$&&$\;\,$$Step^1_{\id{Ann}}$(k)$\inM$S
req $\exists$x.$\;$$\gamma$(enc(k,x))$\;$$=$$\;$$\gamma$(m)
ens token(p)$\;\,$&&$\;\,$$\IOSpecRole{\id{Ann}}$(p,r,S)$\;\,$&&$\;\,$$\exists$x$^\prime$.$\;$m$\;$$=_\Eq$$\;$enc(k,x$^\prime$)
ghost func PaR1(m,p,r,S,k)
\end{gobra}
%\vspace{-0.7em}
\caption{
  Ghost function for the pattern requirement of Example~\ref{ex:message-parsing}. 
  There is the single pattern $\enc{k}{x}$, where $k$ is a constant and $x$ is a variable.
}
\label{fig:pattern-property-contract}
%\vspace{-0.5em}
\end{figure}

\begin{example}[Checking a ciphertext]
\label{ex:message-parsing}
Consider a simple protocol where a role Ann expects \new{a message matching} the pattern $\enc{k}{x}$, where $k$ is a pre-shared key.
We use the fact $Step^1_{\id{Ann}}(k)$ to bind $k$ in the model state.
Figure~\ref{fig:parsing-ciphertext} shows part of an implementation. 
The variable \code{key} stores the pre-shared key, expressed as \code{seq(key)$\;$$=$$\;$$\gamma$(k)}.
After successfully receiving a bytestring~\code{ciph} with a message~\code{c}, \code{seq(ciph)$\;$$=$$\;$$\gamma$(c)} holds due to \code{receive}'s specification (cf.~Figure~\ref{fig:netlib-contracts}). 
Next, the code decrypts \code{ciph}. 
If successful, \code{ciph} equals the \new{bytestring} \code{$\gamma($enc(k,u))} for some message \code{u} with \code{seq(msg)$\;$$=$$\;$$\gamma$(u)} (lines~5--6) 
\new{by} \code{decrypt}'s postcondition (cf.~Figure~\ref{fig:cryptolib-contracts}) and $\gamma$ \new{being} a surjective homomorphism.
Furthermore, we know that \code{$\gamma$(enc(k,u))} equals \code{$\gamma$(c)}, but not yet that the received message \code{c} matches the pattern $\enc{k}{x}$ (line 8).
For this, we apply the pattern requirement by calling the ghost function \code{PaR1} (cf.~Figure~\ref{fig:pattern-property-contract}).
The constant $k$ of the pattern $\enc{k}{x}$ is passed as an argument to the call, and related to the state facts of the I/O specification via the ghost function's precondition (with \code|$Step^1_{\id{Ann}}$(k)$\inM$S|).
\end{example}

\subsubsection{Deriving the pattern requirement}
\label{ssec:deriving-pattern-requirement}

The pattern requirement for a given pattern $t$ can be derived from two more basic properties. \new{We define these properties here and prove this implication. In Section~\ref{ssec:assumption-proof-obligations}, we will discuss assumptions and justifications regarding these properties.}

The first property is \emph{image disjointness}, 
\begin{newtext}% do not remove % here
which has two parts: (i) the images of (public and fresh) names under $\gamma$ are pairwise disjoint and (ii) the image of any function $\Calgify{f}$ for $f \in \Sigma$ neither collides with the image of any other function $\Calgify{g}$, for $g \in \Sigma$, nor with the image of names under $\gamma$.

\begin{definition}[Image disjointness]
\label{def:image-disjointness}
Image disjointness holds \newcr{under the following two conditions:}
\begin{enumerate}[(i)]
%(i) 
\item $\gamma$ is injective on the set of names $\names$ and
%(ii) 
\item for all $f,g \in \Sigma$ such that $f\neq g$, 
\begin{align}
\tag{ID$_{f,g}$}\label{eq:IDf}
\img(\Calgify{f}) \cap (\img(\Calgify{g}) \cup \gamma(\names)) & = \emptyset. 
\end{align}
\end{enumerate}
\end{definition}
\end{newtext}

\begin{newtext}
The second property is \emph{pattern injectivity} for \new{a} pattern~$t$. This constitutes a much weaker form of standard injectivity. It is required to hold only for subterms $t' \sqsubseteq t$ and where, again, equality is guaranteed only modulo a substitution~$\sigma'$.
\begin{definition}[Pattern injectivity] 
\label{def:pattern-injectivity}
Pattern injectivity holds for a pattern $t$ if, for all $f\in \Sigma$ occurring in $t$,
\begin{align}
&f(t'_1, ..., t'_k) \sqsubseteq t \wedge \Calgify{f}(\gamma(t'_1\sigma), ..., \gamma(t'_k\sigma)) = \Calgify{f}(b_1, ..., b_k) \notag \\  
\label{eq:qinject}\tag{PaI$_f$($t$)}
&\implies \exists \sigma'.\: b_1 = \gamma(t'_1\sigma') \wedge ... \wedge b_k = \gamma(t'_k\sigma'). 
\end{align}
\end{definition}
\end{newtext}

\begin{proposition}
\label{prop:pattern-requirement}
Given a linear pattern $t$ (where \new{every variable occurs only once}), image disjointness and pattern injectivity for~$t$ imply the pattern requirement for $t$.
\end{proposition}

\begin{proof}
We prove the proposition's statement for all $t' \sqsubseteq t$ by induction on $t'$.
\end{proof}
\looseness=-1
\new{We split non-linear patterns into multiple linear ones. For instance, the non-linear pattern $t = \tuple{x, \hash{x}}$ can be split into  $t_1 = \tuple{x, \_}$ and $t_2 = \tuple{\_, \hash{x}}$ (where $\_$ matches any term). Conceptually, we then first match a given term $\tuple{u,\hash{u}}$ against $t_1$, which binds $x$ to $u$, and then against $\tuple{\_, \hash{u}}=t_2[x \mapsto u]$. This is equivalent to matching against $t$. This turned out to be simpler to work with than a single linearized pattern with additional equality constraints.
}

\subsubsection{Assumptions and proof obligations}
\label{ssec:assumption-proof-obligations}

\begin{newtext}
We discuss assumptions and proof obligations regarding image disjointness and pattern injectivity. In doing so, we distinguish cryptographic operations from formats. 

Since we are working in a symbolic (Dolev-Yao) model, which assumes perfect cryptography, 
we maintain this assumption for cryptographic operations at the bytestring level in the following form.
\begin{assumption}[Cryptographic operations] We assume that 
\label{assm:crypto-ops}
\begin{enumerate}[(i)]
\item $\gamma$ is injective on the set of names $\names$,
\item \eqref{eq:IDf} holds for cryptographic $f \in \Sigma$ and all $g\in \Sigma$, 
\item \eqref{eq:qinject} holds for all protocol patterns $t$ and all cryptographic $f \in \Sigma$ occurring in $t$.
\end{enumerate}
\end{assumption}
We justify these assumptions by noting that we can expect collisions violating these assumptions to occur only with negligible probability in good cryptographic libraries.
Also recall that pattern injectivity is a much weaker requirement than standard injectivity.

The situation is different for formats (cf.~Section~\ref{ssec:protocol-messages}).  
We can expect that the formats of a well-designed protocol are unambiguously parseable (i.e., injective and hence pattern-injective) and mutually disjoint (i.e., image disjoint). 
We therefore require that these properties are \emph{proved} for formats, e.g., using the techniques proposed in~\cite{DBLP:conf/csfw/ModersheimK14,DBLP:conf/uss/RamananandroDFS19}.
\end{newtext}

\begin{remark}
An obvious way to achieve image disjointness \new{and pattern injectiveness} is to \emph{tag} each construct of the bytestring algebra with a different bytestring. This approach is followed, e.g., in~\cite{Igloo} but it is unrealistic for real protocols.

Alternatively, image disjointness \new{holds} if different operations result in differently sized bytestrings. For operations with varying output sizes, such as stream encryption, this may require restricting \new{the allowed argument sizes in the implementation.} In some cases, this approach may allow us to \emph{prove} image disjointness even for \new{some} cryptographic operations. 
Indeed, \new{we do this for} a pre-existing implementation of the WireGuard protocol, which does not use \new{tagging}.
However, this approach also has its limitations; for example, AES-256 or SHA-256 have the same output size. 
\end{remark}

\subsubsection{\new{Proving term equalities}}
\label{ssec:proving-equalities}
\begin{newtext}
Obligations to prove term equalities during code verification arise from equality constraints in the I/O specification.
However, while the code can check that an equality holds on the bytestring level (e.g., $\gamma(x) = \gamma(\hash{z})$), this does not in general imply the prescribed term equality (e.g., $x=_\Eq\hash{z}$\christoph{Check: All term equalities should be $mod \Eq$.}). 
Following~\cite{DBLP:conf/csfw/DupressoirGJN11,DBLP:journals/jcs/DupressoirGJN14,DBLP:conf/sefm/Vanspauwen015,VanspauwenGijs2017Vcpi}, we can reasonably assume that collisions violating this implication do not occur (with overwhelming probability) in actual protocol executions and thus prove the specification under the condition that bytestring equality implies term equality for the two concrete bytestrings at hand (e.g., $\gamma(x) = \gamma(\hash{z}) \Longrightarrow x=_\Eq\hash{z}$). 
We call this a \emph{collision-freedom} assumption for a given equation.
\end{newtext}

\begin{newtext}
\subsubsection{Summary}
\label{sssec:summary}
The verification of the role implementations, i.e., the Hoare triples $\vdash_{\medextp} \{\PP_{i}(\rid) \}\ \prog_i(rid) \ \{ \mathsf{true} \}$, relies on the following assumptions:
\begin{enumerate}
\item contracts for the I/O and cryptographic libraries, where the former's operations are consistent with~$\medextp$;
\item the pattern requirement for each pattern $t$ occurring in the I/O specification; and
\item collision-freedom for all equalities $\Phi_R$ occurring in the I/O specification.
\end{enumerate}
We suggest to also prove the pattern requirement for a given pattern $t$ whenever possible, e.g., by showing image disjointness and pattern injectivity (at least) for formats and assuming them for the cryptographic operations (Assumption~\ref{assm:crypto-ops}).
\end{newtext}

% !TEX root = main.tex

%%%%%%%%%%%%%%%%%%%%%%%%%%%%%%%%%%%%%%%%%%%%%%%%%%%%%%%%%%%%
\section{\hspace{-1pt}\new{Concrete environment and overall soundness}}
\label{sec:overall-soundness}
%%%%%%%%%%%%%%%%%%%%%%%%%%%%%%%%%%%%%%%%%%%%%%%%%%%%%%%%%%%%

We now derive an overall soundness result for our approach, which relates the abstract Tamarin model to the concrete protocol implementation. 

\subsection{\new{Concrete environment}}
\label{ssec:concrete-environment}
To formulate such a result, we must first define a concrete environment model~$\realEnv$, including a concrete attacker, which can interact with the roles' implementations. These implementations communicate with the environment using I/O library functions, which include non-ghost and ghost parameters: they send and receive both bytestrings (used by the program) and ghost terms (used for the reasoning), related by $\gamma$. The ghost parameters should be reflected in~$\realEnv$'s interface, i.e., its synchronization labels.
Moreover, to fit into an overall soundness result, $\realEnv$ must be trace-included in $\Renv^{e}$.
Hence, the concrete attacker must not be more powerful than the Dolev-Yao attacker, i.e., we must prevent attacks at the bytestring level such as exploiting collisions.

To achieve this, we construct the concrete environment from the term-level environment $\Renv^{e}$ by changing only its \emph{interface} with the protocol. Concretely, we rename and extend every synchronization label $[\lambda_F(rid, \vect{x})]$ of (the LTS induced by) $\Renv^{e}$ to the label $\perm{F}(rid, \new{\gamma(\vect{x}), \vect{x}})$ in $\realEnv$ and we keep the labels of $\Renv^{e}$'s internal actions. Note that applying the relabeling $\medext \circ \medextp$ to $\realEnv$ recovers $\Renv^{e}$'s original labels. Hence, we record the following property.
\begin{proposition}
\label{prop:environment-inclusion}
$\medext(\medextp(\realEnv)) \tracePre \Renv^{e}.$
\end{proposition}

\colorlet{colormed}{blue!75!black}
\colorlet{colorlem}{green!50!blue}
\newcommand{\colormediator}[1]{\textcolor{colormed}{#1}}

\subsection{\new{Overall soundness result}}
\label{ssec:soundness}
Our goal now is to show that any trace property proved for the Tamarin model is preserved
in the concrete system
\[
     \big(\interl_{i,\rid} \medint(\CC_{i}(\rid))\big) \sync{\Lambda'} \realEnv,
\]
which is composed of the \new{verified programs'} LTSs $\CC_{i}$ and the concrete environment $\realEnv$, where \new{the programs' internal operations are mapped back to their action fact arguments} and $\Lambda' = (\medext \circ \medextp)^{-1}(\Lambda)$ synchronizes I/O permissions that also include bytestrings beside terms.
\new{Note that our soundness result assumes that the role implementations are already verified and that the verifier assumption (Assumption~\ref{assm:verifier-assumption}) holds. The verification of the role implementations themselves and the related assumptions are discussed in Section~\ref{sec:implementation}.}
\begin{theorem}[Soundness]
\label{theorem:soundness}
Suppose Assumption~\ref{assm:verifier-assumption} holds 
and that we have verified, for each role $i$, the Hoare triple $\vdash_{\new{\medextp}} \{\PP_{i}(\rid) \}\ \prog_i(rid) \ \{ \mathsf{true} \}$. Then
\[
(\interl_{i,\rid} \new{\medint}(\CC_{i}(\rid))) \sync{\Lambda'} \realEnv \,\tamtracePre\, \msrsys.
\]
\end{theorem}

\begin{figure}[t]
\centering
\begin{tikzpicture}

\tikzstyle{sys}=[node distance = 0.2cm]
\tikzstyle{sys-side}=[node distance = 0]
\tikzstyle{eq}=[rotate=90]
\tikzstyle{lem}=[node distance = 0.1cm, text=colorlem, font=\footnotesize, yshift=-0.02cm]

\node[sys] (A) {$\msrsys$};
\node[sys] (B) [below = 4mm of A] {$\Rintf$};
\node[sys] (C2) [below = 4mm of B] {$\sync{\Lambda}$};
\node[sys-side] (C1) [left = of C2] {$\big(\interl_{i,\rid}\Rrole{i}(\rid)\big)$};
\node[sys-side] (C3) [right = of C2] {$\Renv^e$};
\node[sys] (D) [below = 3mm of C1] {$\colormediator{\med\big(}\PP_{i}(\rid)\colormediator{\big)}$};
\node[sys] (E2) [below = 1.4cm of C2] {$\sync{\Lambda}$};
\node[sys-side] (E1) [left = of E2] {$\Big(\interl_{i,\rid} \colormediator{\med\big(\medextp(}\CC_{i}(\rid)\colormediator{)\big)}\Big)$};
\node[sys-side] (E3) [right = of E2] {$ \colormediator{\medext\big(\medextp(}\realEnv\colormediator{)\big)}$};
\node[sys] (F2) [below = 4mm of E2] {$\sync{\colormediator{\Lambda'}}$};
\node[sys-side] (F1) [left = of F2] {$\interl_{i,\rid} \colormediator{\medint\big(}\CC_{i}(\rid)\colormediator{\big)}$};
\node[sys-side] (F3) [right = of F2] {$\realEnv$};

\path (B) -- node[eq] (I1) {$\tamtracepPre$} (A);
\path (C2) -- node[eq] (I2) {$\tracePre$} (B);
\path (D) -- node[eq] (I3) {$\tracePre$} (C1);
\path (E1) -- node[eq] (I4) [yshift=-0.2cm, xshift=-0.05cm] {$\tracePre$} (D);
\path (E3) -- node[eq] (I5) [yshift=0.3cm] {$\tracePre$} (C3);
\path (F2) -- node[eq, yshift=0.15cm] (I6) {$\tracePre$} (E2);

\node[lem] (R1) [right = of I1.base] {(Lemma~\ref{lem:interface-model})};
\node[lem] (R2) [right = of I2.base] {(Lemma~\ref{lem:decomposition})};
\node[lem] (R3) [right = of I3.base] {(Theorem~\ref{theorem:IO-soundness})};
\node[lem] (R4) [right = of I4.base] {(Assumption~\ref{assm:verifier-assumption})};     
\node[lem] (R5) [right = of I5.base] {(Prop.~\ref{prop:environment-inclusion})};
\end{tikzpicture}

%\vspace{-0.3cm}

\caption{Overview of soundness proof, where $\med$ and $\Lambda'$ are defined by $\med = \medint \circ \medext$ and $\Lambda' = (\medext \circ \medextp)^{-1}(\Lambda)$.}
\label{fig:soundness}
%\vspace{-0.5cm}
\end{figure}

\begin{proof}[Proof]
We decompose the proof into a series of trace inclusions. 
Figure~\ref{fig:soundness} gives an overview of the proof.

The first trace inclusion is
\begin{equation}
\label{eq:direct-composition}
\begin{array}{l}
\phantom{\tracePre}\;  \big(\interl_{i,\rid} \medint(\CC_{i}(\rid))\big) \sync{\Lambda'} \realEnv  \\
\tracePre (\interl_{i,\rid} \med(\medextp(\CC_{i}(\rid)))) \sync{\Lambda} \medext(\medextp(\realEnv)),
\end{array}
\end{equation}
where the first line is \new{obtained from the second by pushing the relabeling $\medext \circ \medextp$ into the parallel composition, thus} changing the set of synchronizing labels from $\Lambda$ to $\Lambda'$. 
Next, we deduce 
\begin{equation}
\label{eq:system-side-inclusion}
\med(\medextp(\CC_{i}(\rid))) \tracePre \Rrole{i}(\rid)
\end{equation}
from Theorem~\ref{theorem:IO-soundness} and from the combination of \new{Assumption~\ref{assm:verifier-assumption} and the assumption $\vdash_{\new{\medextp}} \{\PP_{i}(\rid) \}\ \prog_i(rid) \ \{ \mathsf{true} \}$.}
We then leverage a general composition theorem~\cite[Theorem~2.3]{Igloo} that implies that trace inclusion is  compositional for  a large class of composition operators including $\interl$ and $\sync{\Lambda}$. We apply this to the trace inclusion~\eqref{eq:system-side-inclusion} and the one from Proposition~\ref{prop:environment-inclusion} to derive the trace inclusion 
\begin{equation}
\label{eq:compositional-refinement}
\begin{array}{l}
\phantom{\tracePre}\; (\interl_{i,\rid} \med(\medextp(\CC_{i}(\rid)))) \sync{\Lambda} \medext(\medextp(\realEnv)) \\
\tracePre (\interl_{i,\rid}\Rrole{i}(\rid)) \sync{\Lambda} \Renv^e.
\end{array}
\end{equation}

Our result follows by combining the trace inclusions \eqref{eq:direct-composition} and \eqref{eq:compositional-refinement} with Lemmas~\ref{lem:interface-model} and~\ref{lem:decomposition} and the relation inclusions $\tracePre \:\subseteq\: \tamtracepPre \:\subseteq\: \tamtracePre$ from Section~\ref{sec:background}.
\end{proof}

\begin{corollary}[Property preservation]
\label{cor:property-preservation}
\new{Any trace property $\Phi$ that holds for $\msrsys$ also holds for $(\interl_{i,\rid} \medint(\CC_{i}(\rid))) \sync{\Lambda'} \realEnv$.}
\end{corollary}

\begin{example}[Secrecy for Diffie-Hellman]
\label{example:DH10}
\looseness=-1
We have 
verified 
the \new{secrecy of and agreement on the exchanged key} for the
Tamarin protocol model from Example~\ref{example:DH2}.
\newcr{The former is} expressed as the requirement that in all possible traces of the system,
the attacker never learns a value $k$ for which a $\id{Secret}(k)$ action occurs in the trace.
We then implemented both roles by programs, and 
proved they satisfy their I/O specifications (Example~\ref{example:DH9}).
Our soundness result \new{and its corollary} guarantee that 
the composed system also satisfies \new{key secrecy and authentication}.
\end{example}

% !TEX root = main.tex
\section{\new{Application and WireGuard case study}}
\label{sec:wireguard}

\looseness=-1
To provide evidence that our approach to verifying cryptographic protocol 
implementations is general, powerful, and scales to complex
real-world protocols, we use it to verify \new{our Diffie-Hellman running example and} the WireGuard protocol.
\newcr{Both case studies are available open-source~\cite{artifact}.}

\subsection{\new{Applying our approach}}
\label{ssec:applying-approach}

\begin{newtext}
The application of our approach involves three steps:
\begin{enumerate}
\item \emph{Protocol model specification and verification in Tamarin.} Protocol models must satisfy our mild format restrictions (Section~\ref{sec:formatting}). Existing protocol models may require minor syntactic modifications. Verify the desired trace properties such as secrecy and authentication.

\item \emph{Generation of the roles' I/O specifications.} Our tool automatically generates each protocol role's I/O specification along with definitions of types and internal operations. It accepts all protocol models satisfying our format assumptions. The tool currently supports Gobra (for Go) and VeriFast (for Java).

\item \emph{Role implementation and verification.} Verify an existing or new role implementation against its I/O specification. This relies on user-provided (reusable) contracts for the I/O and cryptographic libraries used and on proofs of the relevant instances of the pattern requirement (e.g., using Assumption~\ref{assm:crypto-ops} and Proposition~\ref{prop:pattern-requirement}). 
\end{enumerate}
By Corollary~\ref{cor:property-preservation} (and Assumption~\ref{assm:verifier-assumption}), all properties proven for the Tamarin model are inherited by the implementation.

\looseness=-1
For our Diffie-Hellman example, we have verified a Go implementation (using Gobra) and a Java implementation (using VeriFast) against their generated specifications.
These two implementations are interoperable and exchange messages via UDP.
We have also produced a faulty implementation that sends~$x$ instead of $g^x$ as the first message and for which verification fails because the I/O permissions do not permit sending this payload.
\end{newtext}

\subsection{The WireGuard key exchange}

\looseness=-1
WireGuard is an open VPN (Virtual Private Network) system that is widely 
deployed on various platforms and integrated into the Linux kernel.
Its core is the WireGuard cryptographic protocol.

\looseness=-1
The WireGuard protocol mainly consists of a handshake, where
two agents establish secret session keys and authenticate each other, and 
a transport phase, where they use these keys to set up a secure channel for 
message transport.
We give an overview of the protocol in Figure~\ref{figure:wireguard}.
The complete protocol additionally features denial of service (DoS) protection mechanisms,
which we omit.

\begin{figure}
\centering
$\begin{array}{ll}
\multicolumn{2}{l}{\textcolor{colorlem}{\text{\it // handshake phase}}}\\
A \rightarrow B: & \tuple{1, sid_I, g^{ek_I}, c_{pk_I}, c_{ts}, mac1_I, mac2_I}\\
B \rightarrow A: & \tuple{2, sid_R, sid_I, g^{ek_R}, c_{empty}, mac1_R, mac2_R}\\
[0.3em]
\multicolumn{2}{l}{\textcolor{colorlem}{\text{\it // transport phase}}}\\
A \rightarrow B: & \tuple{4, sid_R, 0, \aead(\kIR, 0, p_0, '')}\\
B \rightarrow A: & \tuple{4, sid_I, 0, \aead(\kRI, 0, p'_0, '')}\\
A \rightarrow B: & \tuple{4, sid_R, 1, \aead(\kIR, 1, p_1, '')}\\
\dots
\end{array}
$
%\vspace{-.5em}
\caption{The WireGuard protocol.}
\label{figure:wireguard}
%\vspace{-1em}
\end{figure}

\looseness=-1
The protocol involves two roles, the initiator (Alice) and the responder (Bob),
each with long-term private and public keys.
It is assumed that Alice and Bob know each other's public keys $pk_I$ and $pk_R$
in advance. 
They may optionally use a pre-shared secret.
The protocol relies on an authenticated encryption with associated data (AEAD) 
construction $\aead$, and hash and key derivation functions.
The exact algorithms used are irrelevant for our presentation.
All protocol messages contain a tag: 1 and 2 for handshake messages,
3 for the optional DoS prevention messages (not shown), and 4 for 
transport messages. They also contain randomly generated unique 
session identifiers $sid_I$ or $sid_R$ (for each role).

\looseness=-1
The \emph{handshake phase} comprises two messages. Alice and Bob
generate fresh ephemeral Diffie-Hellman keys $ek_I$, $ek_R$, and exchange
the associated public keys $g^{ek_I}$, $g^{ek_R}$.
They also exchange ciphertexts $c_{pk_I}$, $c_{ts}$, and $c_{empty}$,
which respectively encrypt Alice's public key,
a timestamp, and the empty string,
with keys derived from both long term and ephemeral secrets.
The messages also contain message authentication codes (MACs) for the DoS protection mode, 
not described here.
At the end of the handshake, both agents compute two symmetric keys $\kIR$ and $\kRI$. 
The detailed message construction and key computation is given in Appendix~\ref{app:wireguard-message}.

\looseness=-1
These two keys are afterwards used to encrypt messages in the \emph{transport phase}:
$\kIR$ for messages from initiator to responder and $\kRI$ for the other direction.
Alice and Bob both keep two counters $\nIR$ and $\nRI$, counting the number of
messages sent in each direction. 
When Alice sends a message, she encrypts it with $\kIR$, using the current value of $\nIR$ 
as the AEAD nonce, and increments $\nIR$.
Thus, no confusion is possible regarding the order of messages. 
The use of different keys and counters allows
messages to be sent independently in each direction,
without requiring strict alternation.

\looseness=-1
Note that the protocol mandates that Alice sends the first transport message.
The reason is that it is used by Bob to confirm she has received 
his key. In contrast, Alice confirms this for Bob when she receives 
the second handshake message.

\subsection{Tamarin model}

\looseness=-1
We model the WireGuard protocol in Tamarin as a MSR system  
satisfying the assumptions from Section~\ref{sec:formatting}.
Our model features rules for each of the two roles' behavior, as well as 
environment rules modeling the initial long-term key distribution.
The environment may spawn any number of instances of each role,
i.e., an unbounded number of sessions running in parallel,
between the same or different agents.

\looseness=-1
Note that we not only model the handshake and first transport message, after which 
the key exchange is concluded, but also the loop that follows
where agents may exchange any number of transport messages, in any 
order, using the computed keys.
Verifying such unbounded loops is challenging for automated tools,
and often leads to non-termination. For this reason, they are usually not modeled in their full generality.
However, the presence of the loop in the implementation
required its inclusion in the model as well, so that the implementation 
adheres to the model's behavior.
We had to manually write three lemmas and an \emph{oracle} 
(a heuristic for Tamarin's proof search) to help the tool terminate. 
Tamarin can then prove these lemmas and verify the model automatically.

\looseness=-1
We formulate and prove in Tamarin trace properties expressing authentication 
(see Table~\ref{fig:wireguard-properties}).
More precisely, we show that, after the first transport 
message, the participants mutually agree on the resulting keys: 
if Alice believes she has exchanged $\kIR$ and $\kRI$ with Bob, 
then Bob also believes so, and conversely.
Moreover, we prove the forward secrecy of the keys $\kIR$ and $\kRI$:
they remain secret from the attacker, provided neither Alice nor Bob's long-term 
secrets were corrupted \emph{before the end of the key exchange}. %, even if they are corrupted afterwards.
The Tamarin file 
consists of about 250~lines of MSR rules,
and 100~lines of lemmas and properties. It is verified automatically
by Tamarin in about 3 minutes.

\begin{table}[t]
\renewcommand{\arraystretch}{1.1}
 	\caption{Properties verified for the WireGuard case study.
  By Corollary~\ref{cor:property-preservation}, the code inherits the protocol's properties.}
 	\label{fig:wireguard-properties}
	\centering
	\begin{tabular}{l | l} 
 		 Level  & {Verified properties}  \\
		 \hline
 		 Protocol & Agreement on the keys, forward secrecy	\\ 
 		 Code  & Memory safety, conformance with generated I/O spec	\\
 	\end{tabular}
%	\vspace{-1em}
\end{table}

\subsection{Implementation \new{and code verification}}
\label{ssec:implem}
\begin{newtext}\looseness=-1
We separately verified the initiator and responder code of the official Go implementation of WireGuard~\cite{WireGuardSources} in Gobra.
We first proved memory safety, which is independent of our approach, but a required initial step in tools like Gobra and VeriFast. 
We then verified each role implementation against its I/O specification, which we generated with our tool from the WireGuard Tamarin model.
\end{newtext}
We annotated the code with specifications, namely pre- and postconditions and loop invariants, and proof annotations, namely assertions, lemma calls, and predicate unfolding commands. 
The code can be compiled, since the annotations appear inside comments.
\new{Table~\ref{fig:wireguard-properties} summarizes all properties we proved.}

\begin{newtext}
\mypar{Changes to the implementation} \looseness=-1
We modified the official Go implementation in three ways.
First, we removed features not included in our Tamarin model, namely DDoS protection.
Second, we made changes to simplify proving memory safety. For this, we removed metrics and load balancing, which requires complex concurrency reasoning currently not supported by Gobra.
Lastly, we performed changes related to our approach. Namely, %
\end{newtext}%
we wrote stubs for cryptographic and network operations and equipped them with trusted specifications (cf.~Section~\ref{ssec:library-specs}) \new{and adapted the code accordingly}.

\looseness=-1
Our verified implementation is interoperable with the official implementation and can delegate OS traffic over a VPN\@.

\begin{newtext}

\mypar{Verified components} \looseness=-1
We verified both the handshake and transport phase for both roles. 
These components are responsible for all I/O operations of the implementation. 
We did not verify the setup code for network sockets and cryptographic keys. 
The stubs for cryptographic and network operations are trusted by assumption and thus also not verified.

\mypar{I/O specification} \looseness=-1
In addition to the I/O specifications generated by our tool, 
we declared the bytestring algebra operations and the homomorphism $\gamma$.
\end{newtext}

\looseness=-1
\new{Our verification of the I/O specification is standard.}
To verify a call to an I/O operation or an internal operation, we extract the corresponding I/O permission from the predicate $P_i(p,\rid,S)$.
This requires facts about the model state $S$.
For instance, to send a bytestring $b$, there must exist a term~$t$ with $\gamma(t) = b$ such that $\Outf_i(\rid, t) \in S$.
We verify such facts by relating the program state to the model state $S$.

\mypar{Tool expertise requirements}
\newcr{Performing the code-level verification requires a basic understanding of separation logic and tool-specific knowledge to complete the memory-safety proof (as is necessary when verifying any heap-manipulating program).
Justifying each of the program’s I/O operations by relating them to an I/O permission of the I/O specification is mostly straightforward.
}

\begin{newtext}
\mypar{Pattern requirement} \looseness=-1
Each of the protocol's three non-linear patterns induces two instances of the pattern requirement, realized as lemma functions (cf.~Figure~\ref{fig:pattern-property-contract}).

\looseness=-1
We proved the pattern requirement instances using Proposition~\ref{prop:pattern-requirement} and Assumption~\ref{assm:crypto-ops} by showing that all formats are (i) image-disjoint with each other and names and (ii) pattern-injective. 
Point (i) holds, since all formats start with a different constant and their lengths differ from the lengths of bytestrings representing names. Formats are also injective and hence satisfy (ii), as the arguments of all formats appear at fixed offsets in the bytestring representation.
\end{newtext}

\mypar{Statistics} \looseness=-1
Our \newcr{WireGuard} implementation consists of 608~lines of verified Go code, excluding library stubs.
Specifications and proof annotations make up \new{3936}~lines of code, \new{1241 \newcr{(32\%)} of which are generated by our tool.}
We required \new{87}~lines of code to declare the term and bytestring algebras and the axioms about~$\gamma$.
The verification runtime is about \new{148} and \new{138}~seconds for the initiator and responder, respectively.
\newcr{This} case study demonstrates that our approach is applicable to pre-existing real-world security protocols with implementations of considerable size.

\looseness=-1
\newcr{The smaller Diffie-Hellman implementations consist of 106 lines of verified Go code and 113 lines of verified Java code (both excluding library stubs). The Go and Java code require 1015 and 1022 lines of specification of which 623 (61\%) and 780 (76\%) are generated by our tool, respectively.}

% !TEX root = main.tex

\section{Related work}
\label{sec:relatedWork}

We compare our work with different kinds of approaches to formally verifying protocol implementations. \new{We focus on symbolic approaches, as we have already discussed their relation to computational approaches in the introduction.}%\christoph{Is this ok, or too radical?}

\mypar{Model extraction and code generation}
Bhargavan et al.~\cite{DBLP:journals/toplas/BhargavanFGT08} present a sound model extractor from (a first-order subset of) F\# to ProVerif models.
They work with an abstract datatype of bytestrings and corresponding interfaces for the cryptographic and network libraries, which they instantiate both to symbolic terms for prototyping and to actual library implementations.
This approach is used in~\cite{DBLP:journals/tissec/BhargavanFCZ12} to verify TLS~1.0. In~\cite{DBLP:conf/sp/BhargavanBK17}, the authors extract models from a typed JavaScript reference implementation of TLS~1.2 and TLS~1.3.
Several works generate both models for verification and executable code from abstract protocol descriptions. In~\cite{DBLP:conf/fps/Modesti15,DBLP:conf/birthday/AlmousaMV15}, Alice\&Bob-style protocol specifications are translated into ProVerif models and into JavaScript or Java implementations. While in \cite{DBLP:conf/birthday/AlmousaMV15}, the authors prove the correctness of a partial translation from a high-level to a low-level semantics, neither paper proves the full translation's correctness. Sisto et al.~\cite{DBLP:journals/fac/SistoCAP18} 
generate a ProVerif model and a refined Java implementation from an abstract Java protocol specification and prove the implementation's soundness.
We have already discussed the drawbacks of this family of approaches in the introduction.

\mypar{Code verification only}
Bhargavan et al.~\cite{DBLP:conf/popl/BhargavanFG10} modularly verify protocol code written in F\# using the F7 refinement type checker~\cite{DBLP:conf/csfw/BengtsonBFGM08}. They rely on protocol-specific invariants for cryptographic structures, e.g., stating which messages are public. 
Vanspauwen and Jacobs~\cite{DBLP:conf/sefm/Vanspauwen015,VanspauwenGijs2017Vcpi}
use a similar approach for protocols implemented in C and verified using VeriFast. They allow the concrete attacker to directly manipulate bytestrings, which they overapproximate symbolically by a set of terms. However, it is unclear what effect these manipulations have on message parsing in the protocol roles.
While the verification of global protocol properties in the earlier work~\cite{DBLP:conf/popl/BhargavanFG10} required additional hand-written proofs, the more recent work~\cite{BBHHKSW21} enables their verification in a single tool, F$^\star$, by explicitly incorporating a global event trace. 

\looseness=-1
\newcr{While these approaches are also modular and thus scale to protocols like Signal, finding a suitable protocol-specific invariant is challenging. Our approach not only decouples proving the security properties from verifying the implementation’s correctness, it also leverages Tamarin’s automated proof search.}

\mypar{Combined model and code verification}
Dupressoir et al.~\cite{DBLP:conf/csfw/DupressoirGJN11,DBLP:journals/jcs/DupressoirGJN14} use the interactive prover Coq for model verification in combination with the C code verifier VCC. Their approach involves reasoning about concrete bytestrings and their relation to terms. The central definitions of the protocol model are duplicated in Coq and in VCC and some theorems proven in Coq are imported as axioms into VCC. 

\new{Igloo~\cite{Igloo} is a framework for distributed system verification that soundly combines model refinement in Isabelle/HOL with code verification using I/O specifications. Their case studies include a simple authentication protocol. We follow similar steps to extract I/O specifications from Tamarin models,
but do this generically and automatically for a large class of protocol models. In contrast, these steps must be repeated in Igloo for each protocol, which requires Isabelle/HOL expertise. Moreover, our way of relating terms and bytestrings is more flexible and realistic than theirs, which assumes an injective function from bytestrings to terms.
Penninckx et al.~\cite{Penninckx15} introduced I/O specifications for verifying programs' I/O behaviors, but they did not propose a method for verifying global system properties. 
}

\mypar{Message parsing} 
Mödersheim and Katsoris~\cite{DBLP:conf/csfw/ModersheimK14} show that an abstract symbolic model using message formats soundly abstracts a more concrete model, which includes associative message concatenation, variable and fixed length fields, and tags. \new{Their result holds for protocols whose formats are uniquely parseable and image-disjoint and they give algorithms to check these conditions. This work has inspired our use of bytestring algebras and our decomposition of the pattern requirement into image disjointness and pattern injectivity. Their  focus is on abstraction soundness, whereas ours is on code verification.}
EverParse~\cite{DBLP:conf/uss/RamananandroDFS19} is a framework to generate provably secure (i.e., injective and surjective) parsers and serializers for authenticated message formats. 

% !TEX root = main.tex
\section{Conclusion}
\label{sec:conclusions}
We have proposed a novel approach to cryptographic protocol
verification that soundly bridges abstract design models, specified as
multiset rewriting systems, with code-level
specifications. This allows us to leverage the
automation and proof techniques available in Tamarin for design
verification together with state-of-the-art program verifiers
to obtain security guarantees for protocol implementations.
Our approach is general, compatible with different code verification tools,
and applicable to real-world protocols.

There are several exciting directions for future work.
Our framework is based on a Dolev-Yao attacker.
It would be interesting to relax this assumption and allow the concrete attacker to 
perform additional (non-Dolev-Yao) bytestring operations.
Moreover, we currently only support trace properties. 
Some security properties such as privacy properties can be formalized as observational equivalences, which Tamarin also supports~\cite{10.1145/2810103.2813662}.
How to obtain implementation-level equivalence guarantees from a symbolic model is an open problem.

\subsubsection*{Acknowledgements}
We thank the Werner Siemens-Stiftung (WSS) for their generous support of this project.
This work received funding from the France 2030 program managed by the French National Research Agency under grant agreement No.\ ANR-22-PECY-0006. We would also like to thank the anonymous reviewers for their helpful feedback and Sofia Giampietro for her useful comments on an earlier draft of this paper.

%\clearpage

\bibliographystyle{IEEEtran}
\bibliography{references}

\begin{thebibliography}{10}
\providecommand{\url}[1]{#1}
\csname url@samestyle\endcsname
\providecommand{\newblock}{\relax}
\providecommand{\bibinfo}[2]{#2}
\providecommand{\BIBentrySTDinterwordspacing}{\spaceskip=0pt\relax}
\providecommand{\BIBentryALTinterwordstretchfactor}{4}
\providecommand{\BIBentryALTinterwordspacing}{\spaceskip=\fontdimen2\font plus
\BIBentryALTinterwordstretchfactor\fontdimen3\font minus
  \fontdimen4\font\relax}
\providecommand{\BIBforeignlanguage}[2]{{%
\expandafter\ifx\csname l@#1\endcsname\relax
\typeout{** WARNING: IEEEtran.bst: No hyphenation pattern has been}%
\typeout{** loaded for the language `#1'. Using the pattern for}%
\typeout{** the default language instead.}%
\else
\language=\csname l@#1\endcsname
\fi
#2}}
\providecommand{\BIBdecl}{\relax}
\BIBdecl

\bibitem{SchmidtMCB12}
\BIBentryALTinterwordspacing
B.~Schmidt, S.~Meier, C.~J.~F. Cremers, and D.~A. Basin, ``Automated analysis
  of {D}iffie-{H}ellman protocols and advanced security properties,'' in
  \emph{25th {IEEE} Computer Security Foundations Symposium, {CSF} 2012,
  Cambridge, MA, USA, June 25-27, 2012}, 2012, pp. 78--94. [Online]. Available:
  \url{https://doi.org/10.1109/CSF.2012.25}
\BIBentrySTDinterwordspacing

\bibitem{tamarin13}
\BIBentryALTinterwordspacing
S.~Meier, B.~Schmidt, C.~Cremers, and D.~A. Basin, ``The {TAMARIN} prover for
  the symbolic analysis of security protocols,'' in \emph{Computer Aided
  Verification - 25th International Conference, {CAV} 2013, Saint Petersburg,
  Russia, July 13-19, 2013. Proceedings}, 2013, pp. 696--701. [Online].
  Available: \url{https://doi.org/10.1007/978-3-642-39799-8\_48}
\BIBentrySTDinterwordspacing

\bibitem{DBLP:conf/csfw/Blanchet01}
\BIBentryALTinterwordspacing
B.~Blanchet, ``An efficient cryptographic protocol verifier based on {Prolog}
  rules,'' in \emph{14th {IEEE} Computer Security Foundations Workshop
  {(CSFW-14} 2001), 11-13 June 2001, Cape Breton, Nova Scotia, Canada}.\hskip
  1em plus 0.5em minus 0.4em\relax {IEEE} Computer Society, 2001, pp. 82--96.
  [Online]. Available: \url{https://doi.org/10.1109/CSFW.2001.930138}
\BIBentrySTDinterwordspacing

\bibitem{DBLP:conf/sp/BhargavanBK17}
\BIBentryALTinterwordspacing
K.~Bhargavan, B.~Blanchet, and N.~Kobeissi, ``Verified models and reference
  implementations for the {TLS} 1.3 standard candidate,'' in \emph{2017 {IEEE}
  Symposium on Security and Privacy, {SP} 2017, San Jose, CA, USA, May 22-26,
  2017}.\hskip 1em plus 0.5em minus 0.4em\relax {IEEE} Computer Society, 2017,
  pp. 483--502. [Online]. Available: \url{https://doi.org/10.1109/SP.2017.26}
\BIBentrySTDinterwordspacing

\bibitem{DBLP:conf/ccs/BasinDHRSS18}
\BIBentryALTinterwordspacing
D.~A. Basin, J.~Dreier, L.~Hirschi, S.~Radomirovic, R.~Sasse, and V.~Stettler,
  ``A formal analysis of {5G} authentication,'' in \emph{Proceedings of the
  2018 {ACM} {SIGSAC} Conference on Computer and Communications Security, {CCS}
  2018, Toronto, ON, Canada, October 15-19, 2018}, D.~Lie, M.~Mannan,
  M.~Backes, and X.~Wang, Eds.\hskip 1em plus 0.5em minus 0.4em\relax {ACM},
  2018, pp. 1383--1396. [Online]. Available:
  \url{https://doi.org/10.1145/3243734.3243846}
\BIBentrySTDinterwordspacing

\bibitem{DBLP:conf/sp/BasinST21}
\BIBentryALTinterwordspacing
D.~A. Basin, R.~Sasse, and J.~Toro{-}Pozo, ``The {EMV} standard: Break, fix,
  verify,'' in \emph{42nd {IEEE} Symposium on Security and Privacy, {SP} 2021,
  San Francisco, CA, USA, 24-27 May 2021}.\hskip 1em plus 0.5em minus
  0.4em\relax {IEEE}, 2021, pp. 1766--1781. [Online]. Available:
  \url{https://doi.org/10.1109/SP40001.2021.00037}
\BIBentrySTDinterwordspacing

\bibitem{DBLP:conf/fps/Modesti15}
\BIBentryALTinterwordspacing
P.~Modesti, ``{AnBx}: Automatic generation and verification of security
  protocols implementations,'' in \emph{Foundations and Practice of Security -
  8th International Symposium, {FPS} 2015, Clermont-Ferrand, France, October
  26-28, 2015, Revised Selected Papers}, ser. Lecture Notes in Computer
  Science, J.~Garc{\'{\i}}a{-}Alfaro, E.~Kranakis, and G.~Bonfante, Eds., vol.
  9482.\hskip 1em plus 0.5em minus 0.4em\relax Springer, 2015, pp. 156--173.
  [Online]. Available: \url{http://dx.doi.org/10.1007/978-3-319-30303-1_10}
\BIBentrySTDinterwordspacing

\bibitem{DBLP:conf/birthday/AlmousaMV15}
\BIBentryALTinterwordspacing
O.~Almousa, S.~M{\"{o}}dersheim, and L.~Vigan{\`{o}}, ``{Alice} and {Bob}:
  Reconciling formal models and implementation,'' in \emph{Programming
  Languages with Applications to Biology and Security - Essays Dedicated to
  Pierpaolo Degano on the Occasion of His 65th Birthday}, 2015, pp. 66--85.
  [Online]. Available: \url{http://dx.doi.org/10.1007/978-3-319-25527-9_7}
\BIBentrySTDinterwordspacing

\bibitem{DBLP:journals/fac/SistoCAP18}
\BIBentryALTinterwordspacing
R.~Sisto, P.~B. Copet, M.~Avalle, and A.~Pironti, ``Formally sound
  implementations of security protocols with {JavaSPI},'' \emph{Formal Asp.
  Comput.}, vol.~30, no.~2, pp. 279--317, 2018. [Online]. Available:
  \url{https://doi.org/10.1007/s00165-017-0449-8}
\BIBentrySTDinterwordspacing

\bibitem{DBLP:journals/toplas/BhargavanFGT08}
\BIBentryALTinterwordspacing
K.~Bhargavan, C.~Fournet, A.~D. Gordon, and S.~Tse, ``Verified interoperable
  implementations of security protocols,'' \emph{{ACM} Trans. Program. Lang.
  Syst.}, vol.~31, no.~1, pp. 5:1--5:61, 2008. [Online]. Available:
  \url{https://doi.org/10.1145/1452044.1452049}
\BIBentrySTDinterwordspacing

\bibitem{BBHHKSW21}
\BIBentryALTinterwordspacing
K.~Bhargavan, A.~Bichhawat, Q.~H. Do, P.~Hosseyni, R.~Küsters, G.~Schmitz, and
  T.~Würtele, ``{DY*:} {A} modular symbolic verification framework for
  executable cryptographic protocol code,'' in \emph{{IEEE} European Symposium
  on Security and Privacy, EuroS{\&}P 2021, Vienna, Austria, September 6-10,
  2021}.\hskip 1em plus 0.5em minus 0.4em\relax {IEEE}, 2021, pp. 523--542.
  [Online]. Available: \url{https://doi.org/10.1109/EuroSP51992.2021.00042}
\BIBentrySTDinterwordspacing

\bibitem{DBLP:conf/post/CadeB13}
\BIBentryALTinterwordspacing
D.~Cad{\'{e}} and B.~Blanchet, ``Proved generation of implementations from
  computationally secure protocol specifications,'' in \emph{Principles of
  Security and Trust - Second International Conference, {POST} 2013, Held as
  Part of the European Joint Conferences on Theory and Practice of Software,
  {ETAPS} 2013, Rome, Italy, March 16-24, 2013. Proceedings}, ser. Lecture
  Notes in Computer Science, D.~A. Basin and J.~C. Mitchell, Eds., vol.
  7796.\hskip 1em plus 0.5em minus 0.4em\relax Springer, 2013, pp. 63--82.
  [Online]. Available: \url{https://doi.org/10.1007/978-3-642-36830-1\_4}
\BIBentrySTDinterwordspacing

\bibitem{DBLP:conf/sp/Delignat-Lavaud17}
\BIBentryALTinterwordspacing
A.~Delignat{-}Lavaud, C.~Fournet, M.~Kohlweiss, J.~Protzenko, A.~Rastogi,
  N.~Swamy, S.~Z. B{\'{e}}guelin, K.~Bhargavan, J.~Pan, and J.~K. Zinzindohoue,
  ``Implementing and proving the {TLS} 1.3 record layer,'' in \emph{2017 {IEEE}
  Symposium on Security and Privacy, {SP} 2017, San Jose, CA, USA, May 22-26,
  2017}.\hskip 1em plus 0.5em minus 0.4em\relax {IEEE} Computer Society, 2017,
  pp. 463--482. [Online]. Available: \url{https://doi.org/10.1109/SP.2017.58}
\BIBentrySTDinterwordspacing

\bibitem{DBLP:conf/sp/Delignat-Lavaud21}
\BIBentryALTinterwordspacing
A.~Delignat{-}Lavaud, C.~Fournet, B.~Parno, J.~Protzenko, T.~Ramananandro,
  J.~Bosamiya, J.~Lallemand, I.~Rakotonirina, and Y.~Zhou, ``A security model
  and fully verified implementation for the {IETF} {QUIC} record layer,'' in
  \emph{42nd {IEEE} Symposium on Security and Privacy, {SP} 2021, San
  Francisco, CA, USA, 24-27 May 2021}.\hskip 1em plus 0.5em minus 0.4em\relax
  {IEEE}, 2021, pp. 1162--1178. [Online]. Available:
  \url{https://doi.org/10.1109/SP40001.2021.00039}
\BIBentrySTDinterwordspacing

\bibitem{DBLP:conf/lics/Reynolds02}
\BIBentryALTinterwordspacing
J.~C. Reynolds, ``{Separation Logic:} {A} logic for shared mutable data
  structures,'' in \emph{17th {IEEE} Symposium on Logic in Computer Science
  {(LICS} 2002), 22-25 July 2002, Copenhagen, Denmark, Proceedings}.\hskip 1em
  plus 0.5em minus 0.4em\relax {IEEE} Computer Society, 2002, pp. 55--74.
  [Online]. Available: \url{https://doi.org/10.1109/LICS.2002.1029817}
\BIBentrySTDinterwordspacing

\bibitem{Gobra}
\BIBentryALTinterwordspacing
F.~A. Wolf, L.~Arquint, M.~Clochard, W.~Oortwijn, J.~C. Pereira, and
  P.~M{\"{u}}ller, ``{Gobra:} {Modular} specification and verification of {Go}
  programs,'' in \emph{Computer Aided Verification - 33rd International
  Conference, {CAV} 2021, Virtual Event, July 20-23, 2021, Proceedings, Part
  {I}}, ser. Lecture Notes in Computer Science, A.~Silva and K.~R.~M. Leino,
  Eds., vol. 12759.\hskip 1em plus 0.5em minus 0.4em\relax Springer, 2021, pp.
  367--379. [Online]. Available:
  \url{https://doi.org/10.1007/978-3-030-81685-8\_17}
\BIBentrySTDinterwordspacing

\bibitem{DBLP:conf/nfm/JacobsSPVPP11}
\BIBentryALTinterwordspacing
B.~Jacobs, J.~Smans, P.~Philippaerts, F.~Vogels, W.~Penninckx, and F.~Piessens,
  ``{VeriFast:} {A} powerful, sound, predictable, fast verifier for {C} and
  {Java},'' in \emph{{NASA} Formal Methods - Third International Symposium,
  {NFM} 2011, Pasadena, CA, USA, April 18-20, 2011. Proceedings}, ser. Lecture
  Notes in Computer Science, M.~G. Bobaru, K.~Havelund, G.~J. Holzmann, and
  R.~Joshi, Eds., vol. 6617.\hskip 1em plus 0.5em minus 0.4em\relax Springer,
  2011, pp. 41--55. [Online]. Available:
  \url{https://doi.org/10.1007/978-3-642-20398-5\_4}
\BIBentrySTDinterwordspacing

\bibitem{DBLP:conf/cav/Eilers018}
\BIBentryALTinterwordspacing
M.~Eilers and P.~M{\"{u}}ller, ``{Nagini:} {A} static verifier for {Python},''
  in \emph{Computer Aided Verification - 30th International Conference, {CAV}
  2018, Held as Part of the Federated Logic Conference, FloC 2018, Oxford, UK,
  July 14-17, 2018, Proceedings, Part {I}}, ser. Lecture Notes in Computer
  Science, H.~Chockler and G.~Weissenbacher, Eds., vol. 10981.\hskip 1em plus
  0.5em minus 0.4em\relax Springer, 2018, pp. 596--603. [Online]. Available:
  \url{https://doi.org/10.1007/978-3-319-96145-3\_33}
\BIBentrySTDinterwordspacing

\bibitem{Igloo}
\BIBentryALTinterwordspacing
C.~Sprenger, T.~Klenze, M.~Eilers, F.~A. Wolf, P.~M{\"{u}}ller, M.~Clochard,
  and D.~A. Basin, ``{Igloo:} {Soundly} linking compositional refinement and
  separation logic for distributed system verification,'' \emph{Proc. {ACM}
  Program. Lang.}, vol.~4, no. {OOPSLA}, pp. 152:1--152:31, 2020. [Online].
  Available: \url{https://doi.org/10.1145/3428220}
\BIBentrySTDinterwordspacing

\bibitem{artifact}
\BIBentryALTinterwordspacing
L.~Arquint, F.~A. Wolf, J.~Lallemand, R.~Sasse, C.~Sprenger, S.~N. Wiesner,
  D.~Basin, and P.~M{\"{u}}ller, ``Sound verification of security protocols:
  From design to interoperable implementations,'' Aug. 2022, artifact
  containing the specification generation tool and the case studies. [Online].
  Available: \url{https://doi.org/10.5281/zenodo.7409524}
\BIBentrySTDinterwordspacing

\bibitem{CD20195gaka}
\BIBentryALTinterwordspacing
C.~Cremers and M.~Dehnel{-}Wild, ``Component-based formal analysis of {5G-AKA:}
  {Channel} assumptions and session confusion,'' in \emph{26th Annual Network
  and Distributed System Security Symposium, {NDSS} 2019, San Diego,
  California, USA, February 24-27, 2019}.\hskip 1em plus 0.5em minus
  0.4em\relax The Internet Society, 2019. [Online]. Available:
  \url{https://doi.org/10.14722/ndss.2019.23394}
\BIBentrySTDinterwordspacing

\bibitem{CremersHHSM17}
\BIBentryALTinterwordspacing
C.~Cremers, M.~Horvat, J.~Hoyland, S.~Scott, and T.~van~der Merwe, ``A
  comprehensive symbolic analysis of {TLS} 1.3,'' in \emph{Proceedings of the
  2017 {ACM} {SIGSAC} Conference on Computer and Communications Security, {CCS}
  2017, Dallas, TX, USA, October 30 - November 03, 2017}, 2017, pp. 1773--1788.
  [Online]. Available: \url{https://doi.org/10.1145/3133956.3134063}
\BIBentrySTDinterwordspacing

\bibitem{girol-noise-analysis}
\BIBentryALTinterwordspacing
G.~Girol, L.~Hirschi, R.~Sasse, D.~Jackson, C.~Cremers, and D.~Basin, ``A
  spectral analysis of {Noise:} {A} comprehensive, automated, formal analysis
  of {Diffie-Hellman} protocols,'' in \emph{29th {USENIX} Security Symposium
  ({USENIX} Security 20)}.\hskip 1em plus 0.5em minus 0.4em\relax Boston, MA:
  {USENIX} Association, aug 2020. [Online]. Available:
  \url{https://usenix.org/conference/usenixsecurity20/presentation/girol}
\BIBentrySTDinterwordspacing

\bibitem{EMV-brand-mixup}
\BIBentryALTinterwordspacing
D.~Basin, R.~Sasse, and J.~Toro-Pozo, ``Card brand mixup attack: Bypassing the
  {PIN} in non-{Visa} cards by using them for {Visa} transactions,'' in
  \emph{30th {USENIX} Security Symposium ({USENIX} Security 21)}.\hskip 1em
  plus 0.5em minus 0.4em\relax {USENIX} Association, Aug. 2021. [Online].
  Available:
  \url{https://usenix.org/conference/usenixsecurity21/presentation/basin}
\BIBentrySTDinterwordspacing

\bibitem{DolevY83}
\BIBentryALTinterwordspacing
D.~Dolev and A.~C. Yao, ``On the security of public key protocols,''
  \emph{{IEEE} Trans. Information Theory}, vol.~29, no.~2, pp. 198--207, 1983.
  [Online]. Available: \url{https://doi.org/10.1109/TIT.1983.1056650}
\BIBentrySTDinterwordspacing

\bibitem{Lowe97a}
\BIBentryALTinterwordspacing
G.~Lowe, ``A hierarchy of authentication specification,'' in \emph{10th
  Computer Security Foundations Workshop {(CSFW} '97), June 10-12, 1997,
  Rockport, Massachusetts, {USA}}, 1997, pp. 31--44. [Online]. Available:
  \url{https://doi.org/10.1109/CSFW.1997.596782}
\BIBentrySTDinterwordspacing

\bibitem{Penninckx15}
W.~Penninckx, B.~Jacobs, and F.~Piessens, ``Sound, modular and compositional
  verification of the input/output behavior of programs,'' in \emph{Programming
  Languages and Systems}, J.~Vitek, Ed.\hskip 1em plus 0.5em minus 0.4em\relax
  Berlin, Heidelberg: Springer Berlin Heidelberg, 2015, pp. 158--182.

\bibitem{DBLP:conf/csfw/ModersheimK14}
\BIBentryALTinterwordspacing
S.~Mödersheim and G.~Katsoris, ``A sound abstraction of the parsing problem,''
  in \emph{{IEEE} 27th Computer Security Foundations Symposium, {CSF} 2014,
  Vienna, Austria, 19-22 July, 2014}.\hskip 1em plus 0.5em minus 0.4em\relax
  {IEEE} Computer Society, 2014, pp. 259--273. [Online]. Available:
  \url{http://dx.doi.org/10.1109/CSF.2014.26}
\BIBentrySTDinterwordspacing

\bibitem{DBLP:conf/uss/RamananandroDFS19}
\BIBentryALTinterwordspacing
T.~Ramananandro, A.~Delignat{-}Lavaud, C.~Fournet, N.~Swamy, T.~Chajed,
  N.~Kobeissi, and J.~Protzenko, ``{EverParse:} {Verified} secure zero-copy
  parsers for authenticated message formats,'' in \emph{28th {USENIX} Security
  Symposium, {USENIX} Security 2019, Santa Clara, CA, USA, August 14-16, 2019},
  N.~Heninger and P.~Traynor, Eds.\hskip 1em plus 0.5em minus 0.4em\relax
  {USENIX} Association, 2019, pp. 1465--1482. [Online]. Available:
  \url{https://microsoft.com/en-us/research/publication/everparse}
\BIBentrySTDinterwordspacing

\bibitem{DBLP:conf/csfw/DupressoirGJN11}
\BIBentryALTinterwordspacing
F.~Dupressoir, A.~D. Gordon, J.~J{\"{u}}rjens, and D.~A. Naumann, ``Guiding a
  general-purpose {C} verifier to prove cryptographic protocols,'' in
  \emph{Proceedings of the 24th {IEEE} Computer Security Foundations Symposium,
  {CSF} 2011, Cernay-la-Ville, France, 27-29 June, 2011}.\hskip 1em plus 0.5em
  minus 0.4em\relax {IEEE} Computer Society, 2011, pp. 3--17. [Online].
  Available: \url{https://doi.org/10.1109/CSF.2011.8}
\BIBentrySTDinterwordspacing

\bibitem{DBLP:journals/jcs/DupressoirGJN14}
\BIBentryALTinterwordspacing
------, ``Guiding a general-purpose {C} verifier to prove cryptographic
  protocols,'' \emph{Journal of Computer Security}, vol.~22, no.~5, pp.
  823--866, 2014. [Online]. Available: \url{https://doi.org/10.3233/JCS-140508}
\BIBentrySTDinterwordspacing

\bibitem{DBLP:conf/sefm/Vanspauwen015}
\BIBentryALTinterwordspacing
G.~Vanspauwen and B.~Jacobs, ``Verifying protocol implementations by augmenting
  existing cryptographic libraries with specifications,'' in \emph{Software
  Engineering and Formal Methods - 13th International Conference, {SEFM} 2015,
  York, UK, September 7-11, 2015. Proceedings}, ser. Lecture Notes in Computer
  Science, R.~Calinescu and B.~Rumpe, Eds., vol. 9276.\hskip 1em plus 0.5em
  minus 0.4em\relax Springer, 2015, pp. 53--68. [Online]. Available:
  \url{https://doi.org/10.1007/978-3-319-22969-0\_4}
\BIBentrySTDinterwordspacing

\bibitem{VanspauwenGijs2017Vcpi}
\BIBentryALTinterwordspacing
------, ``Verifying cryptographic protocol implementations that use industrial
  cryptographic {APIs},'' Department of Computer Science, KU Leuven, Belgium,
  Tech. Rep., 2017. [Online]. Available:
  \url{https://lirias.kuleuven.be/retrieve/456879}
\BIBentrySTDinterwordspacing

\bibitem{WireGuardSources}
J.~A. Donenfeld, ``Go implementation of {WireGuard},''
  \url{https://git.zx2c4.com/wireguard-go}, [Online; accessed 11-March-2021].

\bibitem{DBLP:journals/tissec/BhargavanFCZ12}
\BIBentryALTinterwordspacing
K.~Bhargavan, C.~Fournet, R.~Corin, and E.~Zalinescu, ``Verified cryptographic
  implementations for {TLS},'' \emph{{ACM} Trans. Inf. Syst. Secur.}, vol.~15,
  no.~1, pp. 3:1--3:32, 2012. [Online]. Available:
  \url{https://doi.org/10.1145/2133375.2133378}
\BIBentrySTDinterwordspacing

\bibitem{DBLP:conf/popl/BhargavanFG10}
\BIBentryALTinterwordspacing
K.~Bhargavan, C.~Fournet, and A.~D. Gordon, ``Modular verification of security
  protocol code by typing,'' in \emph{Proceedings of the 37th {ACM}
  {SIGPLAN-SIGACT} Symposium on Principles of Programming Languages, {POPL}
  2010, Madrid, Spain, January 17-23, 2010}, M.~V. Hermenegildo and
  J.~Palsberg, Eds.\hskip 1em plus 0.5em minus 0.4em\relax {ACM}, 2010, pp.
  445--456. [Online]. Available: \url{https://doi.org/10.1145/1706299.1706350}
\BIBentrySTDinterwordspacing

\bibitem{DBLP:conf/csfw/BengtsonBFGM08}
\BIBentryALTinterwordspacing
J.~Bengtson, K.~Bhargavan, C.~Fournet, A.~D. Gordon, and S.~Maffeis,
  ``Refinement types for secure implementations,'' in \emph{Proceedings of the
  21st {IEEE} Computer Security Foundations Symposium, {CSF} 2008, Pittsburgh,
  Pennsylvania, USA, 23-25 June 2008}.\hskip 1em plus 0.5em minus 0.4em\relax
  {IEEE} Computer Society, 2008, pp. 17--32. [Online]. Available:
  \url{https://doi.org/10.1109/CSF.2008.27}
\BIBentrySTDinterwordspacing

\bibitem{10.1145/2810103.2813662}
\BIBentryALTinterwordspacing
D.~Basin, J.~Dreier, and R.~Sasse, ``Automated symbolic proofs of observational
  equivalence,'' in \emph{Proceedings of the 22nd ACM SIGSAC Conference on
  Computer and Communications Security}, ser. CCS '15.\hskip 1em plus 0.5em
  minus 0.4em\relax New York, NY, USA: Association for Computing Machinery,
  2015, p. 1144–1155. [Online]. Available:
  \url{https://doi.org/10.1145/2810103.2813662}
\BIBentrySTDinterwordspacing

\end{thebibliography}

\appendices
% !TEX root = main.tex

%%%%%%%%%%%%%%%%%%%%%%%%%%%%%%%%%%%%%%%%%%%%%%%%%%%%%%%%%%%%
\section{Protocol format details}
\label{app:protocol-format-details}
%%%%%%%%%%%%%%%%%%%%%%%%%%%%%%%%%%%%%%%%%%%%%%%%%%%%%%%%%%%%

We require that the rules' labels only contain facts from $\sigact$, i.e. for all $\ll\rew{\aa}\rr\in \msrsys$,
$\facts(\aa)\subseteq \sigact$. Here, $\facts(s)$ denotes the set of fact symbols that occur in 
a multiset of facts $s$. 
We assume that environment rules do not directly use the agents' internal states.
Namely, for all $\ll\rew{\aa}\rr\in\Renv$,
$\facts(\ll\cup\rr)\subseteq \sigenv$.
\new{In addition, any rule in $\Renv$ producing a $\Setupf_i$ fact must not produce any other facts on its right-hand side and its label must be empty.}

We also require that rules for role $i$ only use $i$'s state and may consume facts in $\sigin$ 
(but must not produce them) and may produce facts in $\sigout$ (but must not consume them).
More formally, we require, for all $\ll\rew{\aa}\rr\in  \msrsys_i$,
$\facts(\ll) \subseteq \sigstate{i} \cup \sigin$ and $\facts(\rr) \subseteq \sigstate{i} \cup \sigout$.

Finally, we require that for a protocol rule $\ll\rew{\aa}\rr\in \msrsys_i$,
at least one state fact appears in $\rr$, and that there is a $k_i\geq 1$
such that the tuple of the first $k_i$ arguments of all state facts in
$\ll\rew{\aa}\rr$ is the same.  Intuitively, these first $k_i$ arguments
represent parameters of the run of the protocol role: their value
remains fixed throughout the role's execution. They can be,
for instance, the agent's identity, a thread identifier, or any value
that is assumed to be known beforehand by the agent.  We assume that the
first one of these arguments, which we call $\rid$, is of type $\freshtype$.
It is intended to represent a thread identifier.  
For readability, we will usually group these $k_i$ initial parameters as a tuple, denoted by $\vect{init}$.

\new{In summary, these formatting rules only impose very mild constraints on Tamarin models. All protocol models in the Tamarin distribution could easily be adapted to conform to these constraints with only minor modifications. The main changes would be related to providing a separate setup rule for each role $i$ and keeping the arguments of the resulting $\Setupf_i$ fact as the initial arguments of all state facts as described above.}

%%%%%%%%%%%%%%%%%%%%%%%%%%%%%%%%%%%%%%%%%%%%%%%%%%%%%%%%%%%%
\section{Formal definition of parallel compositions}
\label{app:composition-defs}
%%%%%%%%%%%%%%%%%%%%%%%%%%%%%%%%%%%%%%%%%%%%%%%%%%%%%%%%%%%%

\begin{newtext}
We define the (indexed) interleaving parallel composition $|||$ and the (binary) synchronizing parallel composition $\sync{\Lambda}$. These compose their argument MSR systems into a labeled transition system. 

The (indexed) interleaving parallel composition $|||_{i,\rid} \msrsys_i(\rid)$ has as states 
%families of multisets of state facts of the form $f = \{S_{i, \rid}\}_{i, \rid}$ and 
functions $f$ that map each pair $(i, rid)$ to a multiset of state facts and 
transitions $f \trans{\aa} f'$ if, for some $i$ and $rid$, $f(i, rid) \redms{\msrsys_i(\rid)}{\aa} S'$ and $f' = f[(i, \rid) \mapsto S']$, where $f'$ agrees with $f$ except that it maps $(i, \rid)$ to $S'$.

The synchronized composed system $\msrsys_1 \sync{\Lambda} \msrsys_2$ has states of the form $(S_1, S_2)$ and transitions $(S_1, S_2) \trans{\aa} (S_1', S_2')$ if either 
\begin{enumerate}
\item[(i)] $\aa=\emptyM$ and there is an  $\aa' \in_\Eq \Lambda$ such that $S_1\redms{\msrsys_1}{\aa'} S_1'$ and $S_2 \redms{\msrsys_2}{\aa'} S_2'$, 
\item[(ii)] $\aa \notin_\Eq \Lambda$, $S_1\redms{\msrsys_1}{\aa} S_1'$ and $S_2' = S_2$, or 
\item[(iii)] $\aa \notin_\Eq \Lambda$, $S_2 \redms{\msrsys_2}{\aa} S_2'$ and $S_1' = S_1$. 
\end{enumerate}
Here, $\aa' \in_\Eq \Lambda$ means that $\aa' =_\Eq \aa$ for some $\aa \in \Lambda$.

\end{newtext}

%%%%%%%%%%%%%%%%%%%%%%%%%%%%%%%%%%%%%%%%%%%%%%%%%%%%%%%%%%%%
\section{Extended verifier assumption}
\label{app:verifier-assumption}
%%%%%%%%%%%%%%%%%%%%%%%%%%%%%%%%%%%%%%%%%%%%%%%%%%%%%%%%%%%%

\begin{newtext}
Following Penninckx et al.~\cite{Penninckx15}, we sketch an example of semantic assumptions on  programs and Hoare triples that make our extended verifier assumption (Assumption~\ref{assm:verifier-assumption}) hold semantically. The soundness of the program logic itself, i.e., that a provable Hoare triple $\vdash_{\alpha} \{\phi\} \; \prog \; \{\psi\}$ implies its semantic validity $\models_{\alpha} \{\phi\} \; \prog \; \{\psi\}$, is a separate topic beyond the scope of our paper.

We assume that the programming language semantics makes judgements of the form $s, \prog \Downarrow s', \tau$, meaning that the program $\prog$ when started in state $s$ terminates in state $s'$ and produces the I/O trace $\tau$. This semantics induces a LTS $\CC$, whose set of traces for a given starting state $s_0$ is thus 
\[
\traces(\CC) = \{ \tau \mid \exists s'.\, s_0, \prog \Downarrow s', \tau \}.
\]

I/O specifications $\phi$ have both a static and a dynamic semantics, which are defined in terms of \emph{(I/O) heaps}. Heaps are multisets of (ground) I/O permission and token predicates. The static semantics, written $h \models \phi$, intuitively means that $h$ contains (at least) the I/O permissions and tokens prescribed by $\phi$. The dynamic semantics defines the set of traces allowed by an I/O specification $\phi$ to contain those traces that are possible in all heap models of $\phi$, i.e., 
\[
\traces(\phi) = \{ \tau \mid \forall h.\; h \models \phi \Longrightarrow h\trans{\tau}\},
\]
where $h\trans{\tau}$ intuitively means that it is possible to produce a trace~$\tau$ by successively pushing the tokens in $h$ through the I/O permissions in $h$ (and thus consume these permissions).

The semantics of Hoare triples of the form $\{\phi\}\, \prog\, \{\mathsf{true}\}$ with respect to an abstraction function $\alpha$ from program-level I/O operations to abstract I/O permissions is given by
\begin{align*}
& \models_{\alpha} \{\phi\} \; \prog \; \{\mathsf{true}\} \\
& \stackrel{\text{def}}{\Longleftrightarrow} \forall s, \tau, s', h.\; s, \prog \Downarrow s', \tau \,\wedge\, h \models \phi \Longrightarrow h \trans{\alpha(\tau)}  \\
& \Longleftrightarrow \alpha(\traces(\CC)) \subseteq \traces(\phi) \\
& \Longleftrightarrow \alpha(\CC) \tracePre \phi.
\end{align*}
Here, $\alpha(\CC)$ denotes the LTS $\CC$ whose transition labels are renamed under $\alpha$. The final equivalence uses the equality $\alpha(\traces(\CC)) = \traces(\alpha(\CC))$. 

Penninckx et al.~\cite{Penninckx15}'s semantics is formulated for the case where $\alpha$ is the identity function. Both our and their semantics of Hoare triples also include non-trivial post-conditions, which we omit here to simplify the presentation.
\end{newtext}

%%%%%%%%%%%%%%%%%%%%%%%%%%%%%%%%%%%%%%%%%%%%%%%%%%%%%%%%%%%%
\section{WireGuard message construction}
\label{app:wireguard-message}
%%%%%%%%%%%%%%%%%%%%%%%%%%%%%%%%%%%%%%%%%%%%%%%%%%%%%%%%%%%%

The details of the construction of the ciphertexts in the messages for the WireGuard protocol are displayed in Figure~\ref{figure:wireguard-message},
where:
\begin{itemize}
\item $\aead$ is an AEAD algorithm, $\h$ is a hash function, $\kdf_1$, $\kdf_2$, $\kdf_3$ are key derivation functions;
\item \newcr{in $\aead(k, n, p, a)$, $k$ is the key, $n$ a nonce or counter, $p$ the payload (both authenticated and encrypted), and $a$ the additional data (authenticated, but not encrypted);}
\item $(k_I, pk_I)$ and $(k_R, pk_R)$ are the initiator and responder's long-term
private and public keys;
\item $(ek_I, epk_I=g^{ek_I})$ and $(ek_R, epk_R=g^{ek_R})$ are the initiator and 
responder's ephemeral private and public Diffie-Hellman keys, $g$ being the group 
generator; 
\item $\textit{info}$ and $\textit{prologue}$ are fixed strings containing 
protocol information (version, etc.); and
\item $psk$ is an optional pre-shared key -- if unused it is set to a string of zeros.
\end{itemize}

\begin{figure}[!h]
$c_{pk_I}$, $c_{ts}$, $c_{empty}$ are computed as follows.
\[\begin{array}{ll}
    c_0 =& \h(\textit{info})\\
    h_0 =& \h(\tuple{c_0, \textit{prologue}})\\
    h_1 =& \h(\tuple{h_0, pk_R})\\
    c_1 =& \kdf_1(\tuple{c_0, epk_I})\\
    h_2 =& \h(\tuple{h_1, epk_I})\\
    c_2 =& \kdf_1(\tuple{c_1, g^{k_R * ek_I}})\\
    k_1 =& \kdf_2(\tuple{c_1, g^{k_R * ek_I}})\\
c_{pk_I}=& \aead(k_1, 0, pk_I, h_2)\\
    h_3 =& \h(\tuple{h_2, c_{pk_I}})\\
    c_3 =& \kdf_1(\tuple{c_2, g^{k_R * k_I}})\\
    k_2 =& \kdf_2(\tuple{c_2, g^{k_R * k_I}})\\
 c_{ts} =& \aead(k_2, 0, \textit{timestamp}, h_3)\\
    h_4 =& \h(\tuple{h_3, c_{ts}})\\
    c_4 =& \kdf_1(\tuple{c_3,epk_R})\\
    h_5 =& \h(\tuple{h_4,epk_R})\\
    c_5 =& \kdf_1(\tuple{c_4, g^{ek_R * ek_I}})\\
    c_6 =& \kdf_1(\tuple{c_5, g^{ek_R * k_I}})\\
    c_7 =& \kdf_1(\tuple{c_6, psk})\\
    \pi =& \kdf_2(\tuple{c_6, psk})\\
    k_3 =& \kdf_3(\tuple{c_6, psk})\\
    h_6 =& \h(\tuple{h_5, \pi})\\
c_{empty}=&\aead(k_3, 0, "", h_6)\\
\end{array}
\]
$k_{IR}$ and $k_{RI}$ are the resulting exchanged keys: $k_{IR} = \kdf_1(c_7)$, $k_{RI} = \kdf_2(c_7)$.
\caption{The WireGuard message construction}
\label{figure:wireguard-message}
\end{figure}

\ifproofs{%
\newgeometry{left=4cm, right=4cm, bottom=4cm, top=4cm}
\clearpage
\onecolumn
% !TEX root = main.tex
%

\newcommand{\Setup}{\mathit{Setup}}

\newcommand{\eg}{e.g.\xspace}
\newcommand{\ie}{i.e.\xspace}

\newcommand{\R}{\mathscr{R}}
\newcommand{\I}{\mathcal{I}}
\newcommand{\Eio}{\E^{\mathsf{io}}}

\newcommand{\sigkeys}{\sigenv}
\newcommand{\sigev}{\sigact}

\newcommand{\initv}{\vect{init}}

\renewcommand{\tamtraces}{\traces''}

\section{Proofs}
\label{app:proofs}

In the following proofs, we use a slightly different formalisation compared to the description in Section~\ref{sec:theory}.
In that section, we formulated all transformation steps and lemmas in terms
of multiset rewriting (MSR) systems, to make them more legible.
In contrast, in this appendix, we formulate them in terms of the labelled transition systems (LTS) induced by the MSR systems.
This does not change the results, as these are two equivalent descriptions of the same 
system; we simply use the LTS formalism here so that the steps follow more closely the
general Igloo methodology from~\cite{Igloo}.

\subsection{Definitions and notations}
In this section we recall standard definitions for multisets,
as well as the usual Tamarin model and semantics for multiset rewriting,
that were introduced in Section~\ref{sec:background}.

\begin{definition}[Multisets]
A multiset constructed over a set $S$ is a mapping $M: S \rightarrow \N$,
where $M(x)$ represents the number of copies of $x$ in $M$.
We denote $\setM(M) = \{x\in S|M(x) > 0\}$ the set associated to $M$.
We write $x\in M$ if $M(x) > 0$.
We use the notation $\multileft \cdot \multiright$ to define a multiset explicitly: \eg $\multileft a,a,b\multiright$
denotes the multiset containing two copies of $a$ and one of $b$.
In rewriting rules, we alternatively use the notation $[\cdot]$ for the same purpose.

$\cupM$, $\setminusM$ and $\subM$ respectively denote multiset union, difference, and inclusion.
That is, if $A, B$ are two multisets, then for any element $x$,
$(A \cupM B)(x) = A(x)+B(x)$, and $(A\setminusM B)(x) = \max(A(x)-B(x),0)$; and $A \subM B$ when for all $x$, $A(x) \leq B(x)$.

In addition, for a set $S$ and a multiset $M$,
we denote $M\capM S$ the multiset containing only the elements from $M$ (with the same number of copies) that are also in $S$,
\ie $(M \capM S)(x) = M(x)$ if $x\in S$, and $0$ otherwise.
\end{definition}

\medskip

\begin{definition}[Terms]
Consider sets of names $\freshtype$, $\pubtype$, $\vartype$,
representing respectively fresh values, public values, and variables.
Let $\Sigma$ be a finite signature of function symbols, given with their arity.

The set of \emph{terms} constructed from $\Sigma$, $\freshtype$, $\pubtype$, and $\vartype$,
denoted by $\Terms=\Terms_\Sigma(\freshtype\cup\pubtype\vartype)$, is the smallest set containing $\freshtype \cup\pubtype\cup\vartype$ and closed under application of functions in $\Sigma$.

We also let $\GT=\GTerms$ denote the set of \emph{ground terms}, \ie terms without variables.
\end{definition}

\medskip

\begin{definition}[Equational theory]
An equational theory $\Eq$ on $\Terms$ is a set of equations of the form $l = r$,
where $l, r\in\Terms$ are terms without fresh names from $\freshtype$.
We denote $=_\Eq$ equality modulo $\Eq$, \ie the smallest equivalence relation on $\Terms$ that contains $\Eq$, and is stable by 
context and substitution.
\end{definition}

\medskip

\begin{definition}[Facts]
Consider a signature $\sigfacts=\siglin\uplus\sigper$ of \emph{fact symbols} with their arity, partitioned into \emph{linear} facts $\siglin$ and \emph{persistent} facts $\sigper$; as well as a set of terms $\Terms$ with an equational theory $\Eq$.

We write $\fac=\{f(t_1,\dots,t_k)\;|\; f\in \sigfacts \text{ with arity } k, t_1,\dots,t_k\in\Terms\}$ the set of facts 
instantiated with terms, partitioned into $\faclin\uplus\facper$ as expected.

For a set or multiset $s$ of facts, we denote by $\facts(s)$ the subset of $\sigfacts$ containing all fact symbols that
occur in $s$.
\end{definition}

\medskip

\begin{definition}[Multiset rewriting system]
Consider sets of facts $\fac$ and terms $\Terms$ with an equational theory $\Eq$ as defined earlier.
A multiset rewriting system on these terms and facts is a set $\msrsys$ of \emph{rewriting rules} of the form $\ll \rew{\aa} \rr$, where $\ll, \aa, \rr$ are multisets of facts in $\fac$.
We call $\ll$ the premises of the rule, $\rr$ the conclusions, and $\aa$ the events.

The associated transition relation $\redms{\msrsys,\Eq}{\cdot}$ is defined by the rule
\[
 \inferrule
  {\ll \rew{\aa} \rr \in \msrsys\\
   \theta \text{ ground inst. of vars. in } \ll, \aa, \rr\\
   \ll' \rew{\aa'} \rr' =_\Eq (\ll \rew{\aa} \rr)\theta\\
   \ll' \capM \faclin \subM S\\
   \setM(\ll') \cap \facper \subseteq \setM(S)}
  {S \redms{\msrsys,\Eq}{\aa'} S\setminusM (\ll'\capM \faclin) \cupM \rr'}.\]
\end{definition}

We will from now on only consider fact signatures containing reserved fact symbols $\K\in\sigper$, and $\Frf,\Inf,\Outf\in\siglin$.

\medskip

\begin{definition}[Message deduction rules]
We let $\MD_\Sigma$ denote the set of message deduction rules for $\Sigma$:

\[\begin{array}{r@{\;}c@{\;}l}
{[\Outf(x)]}              &\rew{[]}          & {[\K(x)]}\\
{[\K(x)]}                 &\rew{[\knowsf(x)]}& {[\Inf(pk)]}\\
{[]}                      &\rew{[]}          & {[\K(x\in\pubtype)]}\\
{[\Frf(x\in\freshtype)]}  &\rew{[]}          & {[\K(x)]}\\
{[\K(x_1),\dots,\K(x_k)]} &\rew{[]}          & {[\K(f(x_1,\dots,x_k))]} \qquad \text{for } f\in\Sigma \text{ with arity } k\\
\end{array}
\]
\end{definition}

\medskip

\begin{definition}[Freshness rule]
We let $\FreshR$ denote the freshness generation rule:
\[\begin{array}{r@{\;}c@{\;}l}
{[]} & \rew{[\Fre(x)]} & [\Frf(x \in \freshtype)]\\
\end{array}
\]
\end{definition}

\medskip

\begin{definition}[Traces]
The set of traces of a multiset rewriting system $\msrsys$ (for facts $\fac$, terms $\Terms$, and equations $\Eq$) is
\[\begin{array}{l@{}l}
\traces(\msrsys) = \{\tuple{\aa_1,\dots,\aa_m}\;|\;&\exists s_1,\dots,s_m\text{ ground multisets of facts}.\;\\
&\quad \emptyset \redms{\msrsys,\Eq}{\aa_1} s_1 \redms{\msrsys,\Eq}{\aa_2}\dots\redms{\msrsys,\Eq}{\aa_m} s_m\}.
\end{array}\]

The set of filtered traces, without empty labels, is
\[\tamtracesp(\msrsys)=\{\tuple{\aa_i}_{1\leq i \leq m, \aa_i\neq \emptyset} | \tuple{\aa_1,\dots,\aa_m}\in \traces(\msrsys)\}.\]

The set of Tamarin traces, additionally removing collisions in the random generation, is
\[\tamtraces(\msrsys) = \tamtracesp(\msrsys) \setminus \coll,\]
where 
$
\coll = \{\tuple{a_1,\ldots,a_m} \mid \exists i,j.\; i \neq j \land \Frf(n)\in a_i\capM a_j\}.
$
\end{definition}

\medskip

Note that we use the notation $\tuple{x_1,\dots,x_n}$, or alternatively $\tuple{x_i}_{1\leq i \leq n}$, to denote 
the ordered sequence containing the elements $x_1, \dots, x_n$.

\bigskip

A common technique to condition rewrite rules to certain boolean combinations of equalities between messages
is to use equality restrictions. 

\begin{definition}[Equality restrictions]
Assume from now on a set $\faceq\subseteq\siglin$ of fact symbols used to record equalities
(e.g. $\mathsf{Eq}(x,y)$, $\mathsf{NotEq}(x,y)$), that occur only as actions in rules, i.e. in $\aa$ for a rule $\ll \rew{\aa} \rr$.
Consider a mapping $\feq$ from each fact symbol $f\in\faceq$ of arity $k$ to a formula $\feq(f)$ of the form
$\forall x_1, \dots, x_k.\; \phi$ where $\phi$ is a boolean combination of equalities (modulo $\Eq$) between variables $x_1, \dots, x_k$.
Typically, $\feq(\mathsf{Eq}) = \forall x, y.\; x =_\Eq y$, and $\feq(\mathsf{NotEq}) = \forall x, y.\; x \neq_\Eq y$.

The equality restriction associated to $\feq$ is the set of traces
\[
\req = \{\tuple{a_1,\ldots,a_m} \mid \forall i.\; \forall f/k\in\faceq.\;
\forall t_1, \dots, t_k.\; f(t_1,\dots, t_k)\in a_i \Rightarrow \feq(f)(t_1,\dots, t_k)\}.
\]
(abusing notations, we apply $\feq(f)$ to terms to signify instantiating the universally quantified variables with those terms).

Adding this restriction, Tamarin will prove properties of the restricted set of traces
\[\tamtraceseq(\msrsys) = \tamtraces(\msrsys)\cap\req.\]
\end{definition}

\bigskip
We admit the following (easily proved) property, stating that this usual encoding of equalities behaves as expected,
i.e. is equivalent to enforcing the equality conditions at each transition step.

\begin{proposition}
\label{prop:eq-semantics}
Consider the modified semantics $\redmseq{\msrsys,\Eq}{\cdot}$ for the multiset rewriting system $\msrsys$,
that is defined by the rule
\[
 \inferrule
  {\ll \rew{\aa} \rr \in \msrsys\\
   \theta \text{ ground inst. of vars. in } \ll, \aa, \rr\\
   \ll' \rew{\aa'} \rr' =_\Eq (\ll \rew{\aa} \rr)\theta\\
   \ll' \capM \faclin \subM S\\
   \Phieq\\
      \setM(\ll') \cap \facper \subseteq \setM(S)}
  {S \redmseq{\msrsys,\Eq}{\aa'} S\setminusM (\ll'\capM \faclin) \cupM \rr'}\]
  where
  \[\Phieq = \bigwedge_{f\in\faceq \wedge f(t_1,\dots,t_k)\in\aa'}\feq(f)(t_1,\dots,t_k).\]
Let $\tamtracespeq(\msrsys)$ and $\tamtraceseqq(\msrsys)$ be the associated set of traces,
respectively filtered and filtered without collisions, defined as for the normal semantics.
We have
\[\tamtracespeq(\msrsys) = \tamtracesp(\msrsys)\cap\req \quad\text{and}\quad \tamtraceseqq(\msrsys) = \tamtraceseq(\msrsys).\]
\end{proposition}

\bigskip

We prove the refinement results (Lemmas~\ref{lem:interface-model} and~\ref{lem:decomposition})
on the modified semantics for multiset rewriting, which enforces equality checks at each step rather than
in the end as a trace restriction. That is, we prove that the traces of the composition of component systems
are included in the traces of the original multiset rewriting system \emph{for the modified semantics}.

This is more convenient, as it allows us, when writing component systems as guarded event systems, to incorporate the equality 
conditions directly in the guard of each event. That is closer to the way these checks are performed in the implementation,
and therefore helps produce specifications that will be easier to verify on the code.

Therefore, in the end, the global soundness result we obtain states that the traces of the
parallel composition of each role's implementation are included in the traces of the original MSR system
\emph{for the modified semantics}.

Thanks to Proposition~\ref{prop:eq-semantics}, these are the same as the traces for the usual semantics,
with the equality restriction, which is the set of traces considered by Tamarin.

\bigskip
\subsection{Assumptions}
We now recall in detail the formatting assumptions introduced in Section~\ref{sec:formatting}.

\medskip

We fix sets of names $\pubtype$, $\freshtype$, $\vartype$, a signature $\Sigma$,
and the set of terms $\Terms$ constructed on them, with an equational theory $\Eq$.
Let us also fix $n\geq 1$, the number of roles in the system we consider.

\medskip

We consider a fact signature of the form
\[\sigfacts = \sigev \udis \sigkeys \udis (\bigudis_{1\leq i \leq n} \sigstate{i})\]
for some arbitrary (but disjoint) sets of facts symbols $\sigkeys, \sigev, \sigstate{i}$ for $1\leq i \leq n$,
intended to represent respectively environment facts, the action facts, and each role's 
internal states.
Note that each of these sets may contain linear and persistent facts.
We additionally assume that $\sigenv$ contains subsets
$\sigin$, $\sigout$ which we call input and output fact symbols, with $\Frf, \Inf\in\sigin$, $\Outf\in\sigout$, and $\K\in\sigenv\setminus(\sigin\cup\sigout)$.
Furthermore, we assume that there is an initialization fact symbol $\Setupf_i \in \sigin$ for each protocol role $i$. 

\medskip

We then consider a multiset rewriting system $\msrsys$ of the form
\[
\msrsys = \Renv \udis (\bigudis_{1 \leq i \leq n} \msrsys_i).
\]
where $\msrsys_i$ and $\Renv$ are arbitrary (disjoint) sets of rules intended to represent respectively each role
and environment rules, and $\MD_\Sigma\cup\{\FreshR\}\subseteq\Renv$.

\medskip

We assume that for all rule $\ll\rew{\aa}\rr\in\msrsys$, $\facts(\aa)\subseteq\sigact$.
For all environment rule $\ll\rew{\aa}\rr\in\Renv$,
we assume that $\facts(\ll\cup\rr)\subseteq \sigenv$.
In addition, the fact $\Setupf_i$ is only allowed to occur as the only fact produced by rules in $\Renv$.
For all $i$, for all $\ll\rew{\aa}\rr\in \msrsys_i$, we assume:
\begin{itemize}
\item $\facts(\ll)\subseteq\sigin\cup\sigstate{i}$;
\item $\facts(\rr)\subseteq\sigout\cup\sigstate{i}$;
\item at least one state fact appears in $\rr$;
\item the first $k_i\geq 1$ arguments of all state facts in $\ll\rew{\aa}\rr$,
as well as the $\Setupf_i$ fact,
are reserved for role $i$'s parameters, i.e. have the same value in all
state facts and all occurrences of $\Setupf_i$ in the rule.
In addition, the first of these $k$ parameters must be a thread identifier, i.e. a value of
$\freshtype$ called $\rid$.
These $k$ parameters are thus never changed once the $\Setupf_i$ fact is produced to start a 
run of role $i$.
For readability, we will group all these initial parameters as a tuple, denoted by $\initv$.

\end{itemize}

\medskip

For each $\rid\in\freshtype$, and each $i$,
we denote $\msrsys_{i,\rid}$ the set of rules in $\msrsys_i$ 
where the first argument of all state facts is instantiated with $\rid$.

\subsection{Step 1: Interface model}

As stated in Section~\ref{sec:theory},
our goal will be to separate the multiset rewriting system into several transition systems, 
representing each component.
We first introduce interfaces between each component (i.e., each role, and the environment),
to make this separation easier.
These interfaces are additional facts, which we call \emph{buffer} facts.
We add buffers for all operations we wish to consider as I/O, \ie the input and output facts. These take the form of additional rules, called \emph{I/O rules},
that transform each such fact into a buffered version (for $\sigin$) or vice versa (for $\sigout$).

\medskip
We first extend the fact signature by adding, for each input or output fact $F$, 
a ``buffered'' copy $F_{i}$ for each role $i$. We define
\begin{align*}
\sigbuf{i}  & = \{F_{i} \mid F\in\sigin \cup \sigout \}\\
\sigrole{i} & = \sigstate{i} \cup \sigbuf{i} \\
\sigfacts'       & = \sigact \udis \sigenv \udis (\bigudis_{i} \sigrole{i}).
\end{align*}

We then replace the facts used by the protocol rules as follows.
Let $\sigin^-$ be the set of input facts without the role initialisation facts $\Setupf_i$.
For each role $i$, let $\msrsys'_i$ be the set of rules obtained by replacing, in all
rules in $\msrsys_i$, each fact $F(t_1,\dots,t_k)$ (with $F\in\sigin^-\cup\sigout$) with
$F_i(\rid,t_1,\dots,t_k)$, where $\rid$ is the thread id parameter present in the state facts 
in the rule.

We also introduce the set $\Rbuf$ of \emph{I/O rules}, which translate between input or output 
facts and their buffered versions. The set $\Rbuf$ contains the following rules, for each role $i$. 
\begin{align*}
[F(x_1,\dots,x_k)] \xrightarrow{[]} [F_{i}(\rid, x_1,\dots,x_k)]  &&& \text{for $F \in \sigin^-$} \\
[G_{i}(\rid, x_1,\dots,x_k)] \xrightarrow{[]} [G(x_1,\dots,x_k)]  &&& \text{for $G \in \sigout$}
\end{align*}
We also count the role setup rules, which 
generate the $\Setupf_i$ facts, as I/O rules. Hence, we move them from the set $\Renv$ 
to $\Rbuf$, calling the set of remaining environment rules $\Renv^-$.

We then consider the system
\[\Rintf = \Renv^- \udis \Rbuf \udis (\bigudis_{1 \leq i \leq n} \msrsys'_{i}).\]

\bigskip

We can now prove the following lemma, which corresponds to Lemma~\ref{lem:interface-model}
in Section~\ref{ssec:interface-model} (formulated here for the semantics where equality checks are
performed in each transition, rather than as a restriction, which is more convenient later on).
\begin{lemma}
\[\tamtracespeq(\Rintf) \subseteq \tamtracespeq(\msrsys)\]
\end{lemma}

\begin{proof}

We prove this inclusion by establishing a refinement, using a simulation relation $\R$.
That is, we show that
\begin{enumerate}[(1)]
\item $(\emptyset,\emptyset)\in \R$;
\item for all states $(s_1,s'_1)\in \R$, for all transition steps $s'_1 \redmseq{\Rintf,\Eq}{a'} s'_2$
there exists a sequence of transitions $s_1 \redmseq{\msrsys,\Eq}{a_1}\dots\redmseq{\msrsys,\Eq}{a_m} s_m$
such that $(s_m,s'_2)\in \R$, and $\tuple{a_i}_{1\leq i\leq m, a_i\neq \emptyset} = \tuple{a'}$.
In other words, there exists a sequence of transitions that reaches a state related to $s'_2$,
and produces a sequence of $m$ actions, among which one is equal to $a'$, while the others are empty.
In the proof we will actually only have $m\in\{0,1\}$.
\end{enumerate}

We use the relation $\R$ such that $(s,s')\in \R$ if and only if
$s$ is the state obtained from $s'$ by removing the indices $i$ from all facts,
as well as the first argument $\rid$ added to buffered facts.
The first point, $(\emptyset,\emptyset)\in\R$, is clear.

Let $(s_1,s'_1)\in \R$, and consider a step $s'_1 \redmseq{\Rintf,\Eq}{a} s'_2$.
We can distinguish several cases for the rule in $\Rintf$ of which it is an instance.

\begin{itemize}
\item \case{if it is a rule $\ll'\rew{\aa'} \rr'\in \msrsys'_{i}$:} the transition is thus
\[s'_1 \redmseq{\Rintf,\Eq}{\aa'\theta} s'_2 = s'_1 \setminusM (\ll'\theta\capM\faclin) \cupM \rr'\theta\]
for some ground $\theta$ such that $\ll'\theta\capM\faclin \subM s'_1$ and $\setM(\ll'\theta)\cap\facper \subseteq \setM(s'_1)$. Let $\rid$ be the name with which $\theta$
instantiates the first argument (thread id) of the rule's state facts and buffered facts.
By construction, there exists a rule $\ll\rew{\aa}\rr\in \msrsys_{i}$ that can be obtained from $\ll'\rew{\aa'}\rr'$ by
replacing each symbol $F_{i}$ with the unlabelled symbol $F$, and removing its first argument (which is $\rid$).
Note that, by assumption, $\aa$ does not contain fact symbols in $\sigenv$,
and thus $\aa=\aa'$.

Since $(s_1,s'_1)\in\R$ and $\ll'\theta\capM\faclin \subM s'_1$, we have $\ll\theta\capM\faclin \subM s_1$.
For similar reasons, we also have $\setM(\ll\theta)\cap\facper \subseteq \setM(s_1)$.
Hence, by applying rule $(\ll\rew{\aa}\rr)\theta$, we have
\[s_1 \redmseq{\msrsys,\Eq}{\aa\theta} s_1 \setminusM (\ll\theta\capM\faclin) \cupM \rr\theta.\]
Since $(s_1,s'_1)\in\R$, we also have $(s_1 \setminusM (\ll\theta\capM\faclin) \cupM \rr\theta, s'_1 \setminusM (\ll'\theta\capM\faclin) \cupM \rr'\theta)\in\R$.

\item \case{if it is a rule $\ll'\rew{\aa'} \rr'\in \Renv^-$:} the transition is thus
\[s'_1 \redmseq{\Rintf,\Eq}{\aa'\theta} s'_2 = s'_1 \setminusM (\ll'\theta\capM\faclin) \cupM \rr'\theta\]
for some ground $\theta$ such that $\ll'\theta\capM\faclin \subM s'_1$ and $\setM(\ll'\theta)\cap\facper \subseteq \setM(s'_1)$.
By construction, $\ll'\rew{\aa'}\rr'\in\Renv$. By the formatting assumptions,
$\facts(\ll'\cup\rr')\subseteq\sigenv$.
Hence, the facts in $\ll'\theta$ and $\rr'\theta$ are not touched by $\R$.
Thus, since $(s_1,s'_1)\in\R$ and $\ll'\theta\capM\faclin \subM s'_1$,
we have $\ll'\theta\capM\faclin \subM s_1$. Similarly, 
$\setM(\ll'\theta)\cap\facper \subseteq \setM(s_1)$.
Thus, by applying rule $\ll'\rew{\aa'}\rr'$, we have
\[s_1 \redmseq{\msrsys,\Eq}{\aa'\theta} s_1 \setminusM (\ll'\theta\capM\faclin) \cupM \rr'\theta.\]
Since $(s_1,s'_1)\in\R$ and $\facts(\ll'\cup\rr')\subseteq\sigenv$, we also have
$(s_1 \setminusM (\ll'\theta\capM\faclin) \cupM \rr'\theta, s'_1 \setminusM (\ll'\theta\capM\faclin) \cupM \rr'\theta)\in\R$.

\item Finally, the case of I/O rules remains.
We write the proof for the case of an input rule in $\sigin$, as the output case is similar,
and the setup case is similar to the previous case ($\Renv^-$).
Assume the transition rule applied is
$[F(x_1,\dots,x_k)]\rew{[]}[F_i(\rid, x_1,\dots,x_k)]$.
In that case the transition is
\[s'_1=s'\cupM\multileft F(M_1,\dots,M_k)\multiright \redmseq{\Rintf,\Eq}{\emptyset} s_2'=s'\cupM\multileft F_i(\rid,M_1,\dots,M_k)\multiright\]
for some $s'$, $M_1,\dots,M_k$, $\rid$.
Since $(s_1,s'_1)\in\R$, there exists $s$ such that $s_1 = s \cupM \multileft F(M_1,\dots,M_k)\multiright$
and $(s,s')\in\R$.
Hence we also have $(s_1,s'_2)\in\R$, and no transition needs to be performed on $s_1$.
\end{itemize}

\end{proof}

\subsection{Step 2: express the MSR system as an event system}

Before we separate our monolithic system into the composition of several component systems,
we switch to the formalism of guarded event systems, which will be more appropriate when 
extracting an I/O specification.

We associate to $\Rintf$ a guarded transition system $\EE=(\S,\E,\G,\U)$:
\begin{itemize}
\item $\S$ is the set of multisets of ground facts from the signature $\sigfacts'$;
\item $\E = (\bigcup_{1\leq i \leq n, \text{ for each }\rid} \E^s_{i,\rid}) \cup \E^e \cup \Eio$
with
\[\begin{array}{l@{}l}
\E^s_{i,\rid} = \{\skipE\}\cup\{\evs{i,\rid,j}(\theta,\ll',\aa',\rr') \;|\; &
\ll',\aa',\rr' \text{ multisets of ground facts}, \\
&\theta \text{ ground inst. of } \vars(\ll)\cup\vars(\aa)\cup\vars(\rr)\\
&\text{where }
\ll\rew{\aa}\rr \text{ is the $j$th rule in } \msrsys'_{i} \text{ with the first}\\
&\text{parameter of all state and setup facts}\\
&\text{instantiated with } \rid\}
\end{array}\]
and
\[\begin{array}{l@{}l}
\E^e = \{\skipE\}\cup\{\eve{j}(\theta,\ll',\aa',\rr') \;|\; &
\ll',\aa',\rr' \text{ multisets of ground facts}, \\
&\theta \text{ ground inst. of } \vars(\ll)\cup\vars(\aa)\cup\vars(\rr)\\
&\text{where }
\ll\rew{\aa}\rr \text{ is the $j$th rule in } \Renv^-\}\\
\end{array}\]
and
\[\begin{array}{l@{}l}
\Eio = \{\lambda_{F,i,\rid}(\theta,\ll',\aa',\rr') \;|\; &
F\in\sigin\cup\sigout\\
&\ll',\aa',\rr' \text{ multisets of ground facts}, \\
&\theta \text{ ground inst. of } \vars(\ll)\cup\vars(\aa)\cup\vars(\rr)\\
&\text{where }
\ll\rew{\aa}\rr \text{ is the rule associated to fact $F$ and}\\
&\text{role $i$ in $\Rbuf$, with its first parameter instantiated with $\rid$}\}\\
\end{array}\]

\item $\G$: the guard for a parametric event $e(\theta,\ll',\aa',\rr')$ associated to a rule $\ll \rew{\aa} \rr\in \msrsys'$,
in the current state $s$, is
\[\ll'=_\Eq \ll\theta \;\wedge\; \aa'=_\Eq \aa\theta \;\wedge\; \rr'=_\Eq \rr\theta \;\wedge\; G(\ll',s) \;\wedge\;\Phieq(\aa')\]
where
\[G(\ll',s) = (\ll' \capM \faclin \subM s) \;\wedge\; (\setM(\ll')\cap \facper \subseteq \setM(s))\]
is the actual condition for the rule
and
\[\Phieq(\aa') = \bigwedge_{f\in\faceq \wedge f(t_1,\dots,t_k)\in\aa'}\feq(f)(t_1,\dots,t_k)\]
is the equality check.
\item $\U$: a parametric event $e(\theta,\ll',\aa',\rr')$ associated to a rule $\ll \rew{\aa} \rr\in \msrsys'$
updates the state $s$ to $s' = U(\ll',\rr',s)$, where
\[U(\ll',\rr',s) = s \setminusM (\ll'\capM \faclin) \cupM \rr'.\]
\end{itemize}

Note that the event label $\lambda_{F,i,\rid}(\theta,\ll',\aa',\rr')$ we assign to 
events associated with rules in $\Rbuf$ corresponds to the synchronisation label
$\lambda_F(\vect{x})$ described in Section~\ref{ssec:decomposition}.
Since we use the LTS formalism here, it is now an event name, rather than a fact annotating a 
rewrite rule.
We include more parameters here compared to that section to make the form of all our events
uniform -- even though, since the form of the rules in $\Rbuf$ is known, by construction,
it would be sufficient to only give the instantiation of each parameter of $F$, as was done
in the paper.

We let $\redev{\cdot}$ denote the reduction relation defined by checking the guards and applying the updates for each event in $\E$. Formally, $s \redev{E} s'$ if $E$ is an event $e(\theta,\ll',\aa',\rr')$ associated to a rule
$\ll \rew{\aa} \rr\in \Rintf$, such that $\ll'=_\Eq \ll\theta \;\wedge\; \aa'=_\Eq \aa\theta \;\wedge\; \rr'=_\Eq \rr\theta \;\wedge\; G(\ll',s) \;\wedge\;\Phieq(\aa')$, and $s' = U(\ll',\rr',s)$.

\medskip

Let $\pi$ denote the mapping from $\E$ to multisets of facts, that associates to each parametric event
$e(\theta,\ll',\aa',\rr')\in \E$ the multiset of facts $\aa'$.

We denote $\pip$ the extension of $\pi$ to sequences of events and multisets, which applies $\pi$
and removes empty multisets of facts, \ie
for any sequence $e=\tuple{e_1,\dots,e_n}$ of events,
\[\pip(e)=\tuple{\pi(e_i)\;|\; i=1,\dots,n \;\wedge\;\pi(e_i)\neq\emptyset}.\]
Note that Section~\ref{ssec:io-specs} introduced the same mapping, but in the MSR formalism.

\medskip

The traces of the event system are defined as
\[\traces(\EE) = \{e_1,\dots,e_m\;|\; (\forall i.\; e_i\in\E) \;\wedge\; \exists s_1,\dots,s_m\in\S.\; \emptyset \redev{e_1} s_1 \redev{e_2}\dots\redev{e_m} s_m\}.\]

\begin{lemma}
\label{proof-lemma:ev-rintf}
\[\pip(\traces(\EE)) = \tamtracespeq(\Rintf)\]
\end{lemma}
\begin{proof}
We show this lemma by proving that for any state $s$:
\begin{enumerate}[(a)]
\item\label{proof:item1} for any event $e\in\E$ and any $s'$, if $s \redev{e} s'$ then $s \redmseq{\Rintf,\Eq}{\pi(e)} s'$
\item\label{proof:item2} for any reduction step $s \redmseq{\Rintf,\Eq}{\aa} s'$, there exists $e$ such that $s \redev{e} s'$ and $\pi(e)=\aa$.
\end{enumerate}
Once these two properties are established, the lemma follows easily by applying them successively to each reduction step --
\ref{proof:item1} proves the first inclusion $(\subseteq)$ and \ref{proof:item2} proves the second one $(\supseteq)$.

\bigskip

We first show \ref{proof:item1}.
Consider states $s, s'$ and an event $e\in\E$ such that $s \redev{e} s'$.
By construction of $\E$, there exist $e'$, $\theta$, $\ll'$, $\aa'$, $\rr'$ such that $e=e'(\theta,\ll',\aa',\rr')$.
In addition, $e$ is associated to some rule $\ll \rew{\aa} \rr\in \Rintf$.
By definition of the guard associated to $e$, we have $(\ll',\aa',\rr')=_\Eq (\ll,\aa,\rr)\theta$, $\ll'\capM\faclin \subM s$,
the equality conditions $\Phieq(\aa')$ are satisfied,
and $s'=\setM(\ll')\cap\facper\subseteq\setM(s)$.
By definition of the update associated to $e$, we have
$s'=s \setminusM (\ll'\capM \faclin) \cupM \rr'$.
Thus by multiset rewriting, $s \redmseq{\Rintf,\Eq}{\aa'} s'$.
Finally, by definition of $\pi$ we have $\pi(e)=\aa'$, which proves $\ref{proof:item1}$.

\medskip

Let us now show \ref{proof:item2}.
Consider a reduction step $s \redmseq{\Rintf,\Eq}{\aa} s'$.
By definition of multiset rewriting, there exists a rule $\ll\rew{\aa} \rr\in \Rintf$,
a ground instantiation $\theta$ 
of $\vars(\ll)\cup\vars(\aa)\cup\vars(\rr)$, and multisets of ground facts $\ll', \aa', \rr'$
such that $(\ll',\aa',\rr') =_\Eq (\ll,\aa,\rr)\theta$, $\ll'\capM\faclin \subM s$, $\setM(\ll')\cap\facper\subseteq\setM(s)$,
the equality conditions $\Phieq(\aa')$ hold,
and $s'= \setM(\ll')\cap\facper\subseteq\setM(s)$.

By construction of $\Rintf$, $\ll\rew{\aa} \rr$ is either in $\msrsys'_{i}$ for some $i$,
in $\Renv^-$, or in $\Rbuf$.
In either case, there exists a parametric event $e$ such that the family of events 
\[\{e(\theta'',\ll'',\aa'',\rr'')\;|\; \theta'' \text{ gr. inst. of } \vars(\ll)\cup\vars(\aa)\cup\vars(\rr),
\ll'', \aa'', \rr'' \text{ multisets of gr. facts}\}\] is in $\E$.
The three cases differ only by the event name, which is inconsequential here.

Consider the instance $e(\theta,\ll',\aa',\rr')$: its guard is satisfied, and it updates 
$s$ into $\setM(\ll')\cap\facper\subseteq\setM(s)$, \ie $s'$.
Hence $s \redev{e(\theta,\ll',\aa',\rr')} s'$. Since $\pi(e(\theta,\ll',\aa',\rr'))=\aa'$, this proves \ref{proof:item2}.
\end{proof}

From now on, we will actually slightly strengthen the guards of the protocol events in $\E$,
i.e. those in $\E^s_{i,\rid}$, to use true equality instead of equality modulo $\Eq$.
This only reduces the set of possible executions, and thus preserves the previous trace 
inclusion.

\subsection{Step 3: separating the event system into each component}

We are now ready to split the event system into its components.

We will separate it into one system $\EE^s_{i,\rid}$ for each role instance,
and a system $\EE^e$ for the environment.
The interface events, from $\Eio$ (associated to rule in $\Rbuf$), will be split in two
halves, one to be added to $\EE^e$, and the other to one of the $\EE^s_{i,\rid}$.

\bigskip

For each role instance $(i,\rid)$ we consider the event system $\EE^s_{i,\rid}$ containing the events in $\E^s_{i,\rid}$ and the associated 
guards and updates, as well as events, guards and updates encoding the additional rules
\[\begin{array}{rcl@{\quad}l}
{[F_{i}(\rid,\vect{x})]} &\rew{[]}& [] & \text{for $F\in\sigout$}\\
{[]}                  &\rew{[]}& [F_i(\rid,\vect{x})] & \text{for $F\in\sigin$}\\
\end{array}\]
We denote the events for these transitions $\lambda^{s}_{F,i,\rid}$.
Formally:
\begin{itemize}
\item if $F\in\sigout$, the guard for $\lambda^{s}_{F,i,\rid}(\vect{x})$ requires that the state contains $F_{i}(\rid,\vect{x})$,
and its update simply removes that fact from the state or leaves it, depending on whether $F\in\faclin$ or not;
\item if $F\in\sigin$, the guard for $\lambda^{s}_{F,i,\rid}(\vect{x})$ is always satisfied, and its update adds fact $F_{i}(\rid,\vect{x})$
to the state.
\end{itemize}

\bigskip

We also consider the event system $\EE^e$ containing the events in
$\E^e$ and the associated guards and updates, augmented with events encoding the rules
\[\begin{array}{rcl@{\quad}l}
{[F(\vect{x})]}           &\rew{[]}& [] & \text{for $F\in\sigin^-$}\\
{[\vect{F'(\vect{x})}]}                &\rew{[]}& [] & \text{for $F=\Setupf_i$ for some $i$}\\
&&&\text{where $\vect{F'(\vect{x})}$ are the premises of the setup rule for $F$}\\
{[]} &\rew{[]}& [F(\vect{x})] &\text{for $F\in\sigout$.}
\end{array}\]
We denote the events associated to these rules $\lambda^{e}_{F,i,\rid}(\vect{x})$
(i.e. we include one copy of the event for each $i,\rid$).
Formally:
\begin{itemize}
\item if $F\in\sigin^-$, the guard for $\lambda^{e}_{F,i,\rid}(\vect{x})$ requires that the state contains $F(\vect{x})$,
and its update either removes this fact from the state or leaves it, depending on whether $F\in\faclin$ or not;
\item if $F=\Setupf_i$, the guard for $\lambda^{e}_{F,i,\rid}(\vect{x})$ requires that the state contains $[\vect{F'(\vect{x})}]$, and its update removes from the state the facts in this list that are in $\faclin$;
\item if $F\in\sigout$, the guard for $\lambda^{e}_{F,i,\rid}(\vect{x})$ is always satisfied, and its update adds $F(\vect{x})$ to the state.
\end{itemize}

\bigskip

We define the partial function $\chi:(\bigcup_{i,\rid} \EE^s_{i,\rid})\times \EE^e \rightarrow \E$ (conflating $\EE^e$ with its state space, and similarly for $\EE^s$) that synchronises labels $\lambda_F$, i.e.
\begin{itemize}
\item $\chi(\lambda^s_{F,i,\rid}(\vect{m}), \lambda^e_{F,i,\rid}(\vect{x})) = 
\lambda_{F,i,\rid}([\vect{x}\mapsto \vect{m}],\ll',\aa',\rr')$,
where $\ll\rew{\aa}\rr$ is the rule associated with $F$ in $\Rbuf$,
instantiated with $\rid$, $\vect{x}$ are its variables, and $\ll'\rew{\aa'}\rr'$ is its instantiation with $\vect{m}$;
\item $\chi(\skipE,e)=e$ when $e$ is not of the form $\lambda^s_{F,i,\rid}$;
\item $\chi(e,\skipE)=e$ when $e$ is not of the form $\lambda^e_{F,i,\rid}$.
\item $\chi(e,e')$ is undefined in all other cases.
\end{itemize}

\medskip

We introduce two composition operators. $|||$ is parallel composition, i.e. the state space
of $A ||| B$ is the product of those of $A$ and $B$, and its events are $e_A$ updating state 
$(s_A,s_B)$ to $(s_A', s_B)$ if $e_A$ updates $s_A$ to $s_A'$ in $A$, and similarly for $e_B$ in $B$.
${||}_\chi$ is synchronising composition, and is defined similarly:
the state space
of $A ||_\chi B$ is the product of those of $A$ and $B$,
and its events are $\chi(e_A,e_B)$, updating state $(s_A,s_B)$ to $(s_A', s_B')$,
if $e_A$ updates $s_A$ to $s_A'$ in $A$, and similarly for $e_B$ in $B$.

\bigskip

We can then state the following composition lemma,
which corresponds to Lemma~\ref{lem:decomposition} in Section~\ref{ssec:decomposition}
(when chaining it by transitivity with Lemma~\ref{proof-lemma:ev-rintf}).

\begin{lemma}
\[\traces(({|||}_{i,\rid}\EE^s_{i,\rid}) {||}_{\chi} \EE^e) \subseteq \traces(\EE).\]
\end{lemma}
\begin{proof}
We prove this lemma by establishing a refinement with a simulation relation $\R$ between abstract states (from $\EE$)
and concrete ones (from $({|||}_{i}\EE^s_{i,\rid}) {||}_{\chi} \EE^e$).
The concrete states are of the form $((s_{i,\rid})_{1\leq i \leq n, \text{ for each }\rid},s_e)$,
\ie they are composed of one multiset of facts for each $i,\rid$, and one for the environment.

The refinement relation we use is defined by $(s,s')\in\R$ iff $s = (\cupM_{i,\rid} s'_{i,\rid}) \cupM s'_e$,
where $s'=((s'_{i,\rid}),s'_e)$.

For better legibility we will denote $\EE'=({|||}_{i}\EE^s_{i,\rid}) {||}_{\chi} \EE^e$ the composed system,
and $\redev{\cdot}_{\EE'}$, $\redev{\cdot}_{\EE}$, $\redev{\cdot}_{i,\rid}$, $\redev{\cdot}_{e}$
the transition relations defined respectively by $\EE'$, $\EE$, $\EE^s_{i,\rid}$ and $\EE^e$.

It is clear that the initial states of $\EE'$ and $\EE$ are related:
$(((\emptyset,\dots,\emptyset),\emptyset),\emptyset)\in\R$. We now show that for 
all states $(s_1,s'_1)\in \R$, for all transition steps $s'_1 \redev{e}_{\EE'} s'_2$
there exists a transition $s_1 \redev{e}_{\EE} s_2$
such that $(s_2,s'_2)\in \R$.
We will denote $s'_j=((s'_{j,i,\rid})_{1\leq i \leq n, \text{ for each }\rid},s'_{j,e})$ for $j\in\{1,2\}$.

Following the definition of $\chi$, we can distinguish several cases for the transition step $s'_1 \redev{e}_{\EE'} s'_2$.
\begin{itemize}
\item \case{if $e=\chi(e',\skipE)$ for some $e'\in \E^s_{i,\rid}$:} then $e'=e$,
and by definition of $||_\chi$ and $|||$, since $s'_1 \redev{e}_{\EE'} s'_2$,
we have $s'_{1,i,\rid}\redev{e}_{i,\rid} s'_{2,i,\rid}$,
$s'_{2,j,\rid'}=s'_{1,j,\rid'}$ for all $(j,\rid')\neq (i,\rid)$, and $s'_{2,e}=s'_{1,e}$.

By definition of $\EE^s_{i,\rid}$, the guard and update of $e$ in that system are the same as in $\EE$.
In addition, since $(s'_1,s_1)\in\R$, we have $s'_{1,i,\rid}\subM s_1$.
It is immediate from the form of the guard in $\EE$ that it is stable by supermultiset, and thus holds for $s_1$.
It is also clear from the form of the update in $\EE$ that applying it to the whole state is exactly the same as
applying it to a submultiset of the state
that satisfies the guard, here $s'_{1,i,\rid}$, and leaving the rest of the state untouched.

Therefore, we have
\[s_1=(\cupM_{j,\rid'} s'_{1,j,\rid'}) \cupM s'_{1,e} \;\redev{e}_{\EE}\;(\cupM_{(j,\rid')\neq(i,\rid)} s'_{1,j,\rid'})\cupM s'_{2,i,\rid} \cupM s'_{1,e}.\]

As noted earlier, we have $s'_{2,e}=s'_{1,e}$ and $s'_{2,j,\rid'}=s'_{1,j,\rid'}$ for all $(j,\rid')\neq (i,\rid)$,
hence $(s'_2,(\cupM_{(j,\rid')\neq(i,\rid)} s'_{1,j,\rid'})\cupM s'_{2,i,\rid} \cupM s'_{1,e})\in\R$,
which concludes the proof in this case.

\item \case{if $e=\chi(\skipE,e')$ for some $e'\in \E^e$:} this case is similar to the previous one.

\item The remaining case is \case{the synchronisation case,
where $e=\chi(\lambda^s_{F,i,\rid}(\vect{m}), \lambda^e_{F,i,\rid}(\vect{m}))$} {for some $F, i, \rid, \vect{m}$.}

Then $e=\lambda_{F,i,\rid}(\theta_m,\ll',\aa',\rr')$
where $\theta_m=[\vect{x}\mapsto \vect{m}]$, $\ll\rew{\aa}\rr$ is the rule associated with $F$ in $\Rbuf$,
instantiated with $\rid$, $\vect{x}$ are its variables, and $\ll'\rew{\aa'}\rr'$ is its instantiation with $\theta_m$.

By definition of $||_\chi$ and $|||$, since $s'_1 \redev{e}_{\EE'} s'_2$,
we have $s'_{1,i,\rid}\redev{\lambda^{s}_{F,i,\rid}(\vect{m})}_{i,\rid} s'_{2,i,\rid}$,
$s'_{1,e}\redev{e^{e,in}(m)}_{e} s'_{2,e}$, and
$s'_{2,j,\rid'}=s'_{1,j,\rid'}$ for all $(j,\rid')\neq (i,\rid)$.

We now need to distinguish the cases where $F\in\sigin^-$, $F=\Setup_i$, and $F\in\sigout$.
We write the proof for the case $F\in\sigout$ and $F\in\faclin$,
as the other cases are similar.

We then have $\ll'=\multileft F_i(\rid,\vect{m})\multiright$, $\aa'=[]$, and
$\rr'=\multileft F(\vect{m})\multiright$.

By definition of $\lambda^{s}_{F,i,\rid}(\vect{m})$, we have 
$F_i(\rid,\vect{m})\in $
$s'_{2,i,\rid}=s'_{1,i,\rid}$, and
$s'_{1,i,\rid}\setminusM\multileft F_i(\rid,\vect{m}) \multiright$.

By definition of $\lambda^e_{F,i,\rid}(\vect{m})$, we have $s'_{2,e}=s'_{1,e}\cupM\multileft F(\vect{m})\multiright$.
Therefore, $s_1=(\cupM_{k,\rid'} s'_{1,k,\rid'}) \cupM s'_{1,e}$ satisfies the guard of the event $\lambda_{F,i,\rid}(\theta_m,\multileft F_i(\rid,\vect{m})\multiright,\emptyset,\multileft F(\vect{m})\multiright)$
associated with rule $[F_i(\rid,\vect{m})] \rightarrow [F(\vect{m})]$ in $\EE$.

Hence we have
\[\begin{array}{l}
s_1 \;\redev{\lambda_{F,i,\rid}(\theta_m,\multileft F_i(\rid,\vect{m})\multiright,\emptyset,\multileft F(\vect{m})\multiright)}_{\EE}\\
\qquad\qquad((\cupM_{(k,\rid')} s'_{1,k,\rid'}) \cupM s'_{1,e}\setminusM\multileft F_i(\rid,\vect{m}) \multiright)\cupM\multileft F(\vect{m}) \multiright.
\end{array}\]
As noted earlier, we have $s'_{2,i,\rid}=s'_{1,i,\rid} \cupM \multileft  F(\vect{m})\multiright$,
$s'_{2,e}=s'_{1,e}\setminusM\multileft F_i(\rid,\vect{m}) \multiright$, and $s'_{2,k,\rid'}=s'_{1,k,\rid'}$ for all $(k,\rid')\neq (i,\rid)$.
Thus $(s'_2,((\cupM_{(k,\rid')} s'_{1,k,\rid'}) \cupM s'_{1,e}\setminusM\multileft F_i(\rid,\vect{m}) \multiright)\cupM\multileft F(\vect{m})\multiright)\in\R$,
which concludes the proof in this case.
\end{itemize}

\end{proof}
}{}

\end{document}